\documentclass[11pt]{article}

\usepackage[margin=1in]{geometry}
\usepackage{amsfonts,amssymb,amsmath,amsthm,mathtools,mathrsfs}
\usepackage{thmtools,thm-restate,wrapfig,latexsym}

\usepackage{graphicx,microtype,lpic,bm,xspace,float}
\usepackage[bf]{caption}

\usepackage{enumitem}
\setlist[itemize]{noitemsep,listparindent=\parindent}	
\setlist[enumerate]{noitemsep}

\usepackage[usenames,dvipsnames]{xcolor}
\usepackage{hyperref}
\hypersetup{
linktoc=all,
	colorlinks=true,
	linkcolor=Sepia,
	citecolor=Sepia,
	filecolor=Sepia,
	urlcolor=Sepia
}

\usepackage{tikz}
\usetikzlibrary{decorations.pathmorphing}

\usepackage[framemethod=tikz]{mdframed}
\mdfsetup{
	roundcorner=4pt,
	nobreak=true,
	linewidth=.5,
	innerleftmargin=20pt,
	innerrightmargin=20pt,
	innertopmargin=10pt,
	innerbottommargin=10pt,
}
\newenvironment{myframe}[2]
  {\mdfsetup{
		frametitle={\hspace{#2\textwidth}#1},
		skipabove=3mm,
		skipbelow=2mm,
		leftmargin=2.2cm,
		rightmargin=2.2cm,
    innertopmargin=2mm,
		innerbottommargin=7pt,
		innerleftmargin=10pt,
		innerrightmargin=10pt,
		linecolor=black,
		linewidth=.5pt
    }
  \begin{mdframed}%
  }
  {\end{mdframed}}

\declaretheorem[name=Theorem]{theorem}
\declaretheorem[name=Lemma,sibling=theorem]{lemma}
\declaretheorem[name=Claim,sibling=theorem]{claim}
\declaretheorem[name=Corollary,sibling=theorem]{corollary}
\declaretheorem[name=Fact,sibling=theorem]{fact}
\declaretheorem[name=Definition,numberlike=theorem,style=definition]{definition}

\declaretheorem[name={1-vs-3~Lemma},numbered=no]{one-vs-three}
\declaretheorem[name={3-vs-7~Lemma},numbered=no]{three-vs-seven}
\declaretheorem[name={Homogeneity~Lemma},numbered=no]{homogeneity}
\newcommand{\onethreename}{1-vs-3 Lemma}
\newcommand{\threesevenname}{3-vs-7 Lemma}
\newcommand{\homogname}{Homogeneity Lemma}
\newcommand{\onethree}{\hyperref[thm:1v3]{\onethreename}\xspace}
\newcommand{\threeseven}{\hyperref[thm:3v7]{\threesevenname}\xspace}
\newcommand{\homog}{\hyperref[thm:homog]{\homogname}\xspace}

\newcommand{\Z}{\mathbb{Z}}
\newcommand{\R}{\mathbb{R}}

\newcommand{\I}{\mathbb{I}}
\renewcommand{\H}{\mathbb{H}}
\newcommand{\D}{\mathbb{D}}
\renewcommand{\Pr}{\mathbb{P}}
\newcommand{\E}{\mathbb{E}}
\newcommand{\semi}{{}\operatorname{;}{}}
\renewcommand{\mid}{\,|\,}
\newcommand{\bigmid}{\,\big|\,}

\newcommand{\midd}{\,\|\,}

\newcommand{\calG}{\mathcal{G}}

\newcommand{\calP}{\mathcal{P}}
\newcommand{\calQ}{\mathcal{Q}}

\newcommand{\calS}{\mathcal{S}}

\newcommand{\calX}{\mathcal{X}}
\newcommand{\calY}{\mathcal{Y}}

\newcommand{\witness}{$(\#\exists\!-\!1)$}

\newcommand{\KWEF}{KW\!/EF\xspace}
\newcommand{\KW}{\textsc{kw}}
\newcommand{\A}{\textrm{A}}
\newcommand{\B}{\textrm{B}}

\newcommand{\smallFunction}[2]{\newcommand{#1}{{\textsc{#2}}}}
\smallFunction{\AND}{And}
\smallFunction{\NAND}{Nand}
\smallFunction{\OR}{Or}
\smallFunction{\XOR}{Xor}
\smallFunction{\IP}{IP}
\smallFunction{\TSE}{Tse}
\smallFunction{\SAT}{Sat}

\DeclareMathOperator{\xc}{xc}
\DeclareMathOperator{\rk}{rk}

\DeclareMathOperator{\var}{var}
\DeclareMathOperator{\conv}{conv}
\DeclareMathOperator{\viol}{viol}
\DeclareMathOperator{\acc}{acc}
\DeclareMathOperator{\KWp}{\textrm{\upshape KW$^{+}$}}

\newcommand{\Pia}{\Pi^{\acc}}

\renewcommand{\nwarrow}{\text{\rotatebox{135}{$\rightarrow$}}\!}

\begin{document}

\newgeometry{margin=1.5in,top=1.8in}

\begin{center}
{\LARGE Extension Complexity of Independent Set Polytopes}
\\[1cm] \large
\begin{tabular}{c@{\hspace{1.3cm}}c@{\hspace{1.3cm}}c}
Mika G\"o\"os$^1$ &
Rahul Jain$^2$ &
Thomas Watson$^1$
\end{tabular}

\vspace{5mm}

{\slshape\small
$^1$Department of Computer Science, University of Toronto \\[0mm]
$^2$Centre for Quantum Technologies and Department of Computer Science,\\[-0.5mm]
National University of Singapore and MajuLab, UMI 3654, Singapore}

\vspace{9mm}

\today

\vspace{9mm}

\normalsize
\bf Abstract
\end{center}

\noindent
We exhibit an $n$-node graph whose independent set polytope requires extended formulations of size exponential in $\Omega(n/\log n)$. Previously, no explicit examples of $n$-dimensional $0/1$-polytopes were known with extension complexity larger than exponential in $\Theta(\sqrt{n})$. Our construction is inspired by a relatively little-known connection between extended formulations and (monotone) circuit depth.

\thispagestyle{empty}
\newpage
\newgeometry{margin=1.2in,top=1.5in}
\thispagestyle{empty}
\setcounter{page}{0}

\tableofcontents

\restoregeometry
\newpage
\section{Introduction} \label{sec:intro}

A polytope $P\subseteq \R^n$ with many facets can sometimes admit a concise description as the projection of a higher dimensional polytope $E\subseteq\R^e$ with few facets. This phenomenon is studied in the theory of ``extended formulations''. The \emph{extension complexity} $\xc(P)$ of a polytope $P$ is defined as the minimum number of facets in any $E$ (called an \emph{extended formulation} for $P$) such that
\[
P~=~\{x\in\R^n: (x,y)\in E\text{ for some }y\}.
\]
Extended formulations are useful for solving combinatorial optimization problems: instead of optimizing a linear function over $P$, we can optimize it over $E$---this may be more efficient since the runtime of~LP solvers often depends on the number of facets.

Fiorini et al.~\cite{fiorini15exponential} were the first to show (using methods from communication complexity~\cite{kushilevitz97communication,jukna12boolean}) exponential extension complexity lower bounds for many explicit polytopes of relevance to combinatorial optimization, thereby solving an old challenge set by Yannakakis~\cite{yannakakis91expressing}. For example, their results include a $2^{\Omega(m)}$ lower bound for the $\binom{m}{2}$-dimensional \emph{correlation/cut polytope}. In another breakthrough, Rothvo\ss~\cite{rothvos14matching} proved a much-conjectured $2^{\Omega(m)}$ lower bound for the $\binom{m}{2}$-dimensional \emph{matching polytope}. By now, many accessible introductions to extended formulations are available; e.g., Roughgarden~\cite[\S5]{roughgarden15communication}, Kaibel~\cite{kaibel11extended}, Conforty~et~al.~\cite{conforti10extended} or their textbook~\cite[\S4.10]{conforti14integer}.

\paragraph{\boldmath $\sqrt{n}$-frontier.}
Both of the results quoted above---while optimal for their respective polytopes---seem to get ``stuck'' at being exponential in the square root of their dimension. In fact, no explicit $n$-dimensional 0/1-polytope (convex hull of a subset of $\{0,1\}^n$) was known with extension complexity asymptotically larger than $2^{\Theta(\sqrt{n})}$. In comparison, Rothvo\ss~\cite{rothvos12some} showed via a counting argument that most $n$-dimensional 0/1-polytopes have extension complexity~$2^{\Omega(n)}$.

\subsection{Our result}
Our main result is to construct an explicit 0/1-polytope of near-maximal extension complexity~$2^{\Omega(n/\log n)}$. Moreover, the polytope can be taken to be the \emph{independent set polytope} $P_G$ of an $n$-node graph $G$, i.e., the convex hull of (the indicator vectors of) the independent sets of $G$. Previously, a lower bound of $2^{\Omega(\sqrt{n})}$ was known for independent set polytopes~\cite{fiorini15exponential}.
\begin{theorem} \label{thm:main}
There is an (explicit) family of $n$-node graphs $G$ with $\xc(P_G)\geq 2^{\Omega(n/\log n)}$.
\end{theorem}

In fact, our graph family has bounded degree. Hence, using known reductions, we get as a corollary quantitative improvements---from $2^{\Omega(\sqrt{n})}$ to $2^{\Omega(n/\log n)}$---for the extension complexity of, for instance, \emph{3SAT} and \emph{knapsack polytopes}; see~\cite{avis14extension,pokutta13note} for details.

We strongly conjecture that our graph family actually satisfies $\xc(P_G)\geq 2^{\Omega(n)}$, i.e., that the $\log n$ factor in the exponent is an artifact of our proof technique. We give concrete evidence for this by proving an optimal bound for a certain \emph{query complexity} analogue of \autoref{thm:main}. In particular, the conjectured bound $\xc(P_G)\geq 2^{\Omega(n)}$ would follow from quantitative improvements to the known query-to-communication simulation theorems (\cite{goos15rectangles} in particular). Incidentally, this also answers a question of Lov{\'a}sz, Naor, Newman, and Wigderson~\cite{lovasz95search}: we obtain a maximal $\Omega(n)$ lower bound on the randomized query complexity of a search problem with constant certificate complexity.

\subsection{Our approach} \label{sec:approach}

Curiously enough, an analogous $\sqrt{n}$-frontier existed in the seemingly unrelated field of \emph{monotone circuits}: Raz and Wigderson~\cite{raz92monotone} proved an $\Omega(m)$ lower bound for the depth of any monotone circuit computing the \emph{matching function} on $\binom{m}{2}$ input bits. This remained the largest monotone depth bound for an explicit function until the recent work of G\"o\"os and Pitassi~\cite{goos14communication}, who exhibited a function with monotone depth $\Omega(n/\log n)$. In short, our idea is to prove an extension complexity analogue of this latter result.

The conceptual inspiration for our construction is a relatively little-known connection between Karchmer--Wigderson games~\cite{karchmer88monotone} (which characterize circuit depth) and extended formulations. This ``\KWEF connection'' (see \autoref{sec:kw-ef} for details) was pointed out by Hrube\v{s}~\cite{hrubes12nonnegative} as a nonnegative analogue of a classic rank-based method of Razborov~\cite{razborov90applications}. In this work, we focus only on the monotone setting. For any monotone $f\colon\{0,1\}^n\to\{0,1\}$ we can study the convex hull of its $1$-inputs, namely, the polytope
\[
F~\coloneqq~\conv f^{-1}(1).
\]
The upshot of the \KWEF connection is that extension complexity lower bounds for $F$ follow from a certain type of \emph{strengthening} of monotone depth lower bounds for $f$. For example, using this connection, it turns out that Rothvo\ss's result~\cite{rothvos14matching} implies the result of Raz and Wigderson~\cite{raz92monotone} in a simple black-box fashion (\autoref{sec:matching}).

Our main technical result is to strengthen the existing monotone depth lower bound from~\cite{goos14communication} into a lower bound for the associated polytope (though we employ substantially different techniques than were used in that paper). The key communication search problem studied in \cite{goos14communication} is a communication version of the well-known \emph{Tseitin} problem (see \autoref{sec:tseitin-def} for definitions), which has especially deep roots in proof complexity (e.g., \cite[\S18.7]{jukna12boolean}) and has also been studied in query complexity~\cite{lovasz95search}. We use information complexity techniques to prove the required $\Omega(n/\log n)$ communication lower bound for the relevant variant of the Tseitin problem; information theoretic tools have been used in extension complexity several times~\cite{braverman13information,braun13common,braun15matching}. One relevant work is Huynh and Nordstr\"om~\cite{huynh12virtue} (predecessor to \cite{goos14communication}), whose information complexity arguments we extend in this work.

(Instead of using information complexity, an alternative seemingly promising approach would be to ``lift'' a strong enough query complexity lower bound for Tseitin into communication complexity. Unfortunately, this approach runs into problems due to limitations in existing query-to-communication simulation theorems; we discuss this in \autoref{sec:query-lb}.)

\autoref{thm:main} follows by reductions from the result for Tseitin (\autoref{sec:reductions}). Indeed, it was known that the Tseitin problem reduces to the monotone KW game associated with an $f\colon\{0,1\}^{O(n)}\to\{0,1\}$ that encodes (in a monotone fashion) a certain CSP satisfiability problem. This gives us an extension complexity lower bound for the (explicit) polytope $F \coloneqq \conv f^{-1}(1)$. As a final step, we give a reduction from $F$ to an independent set polytope.

\subsection{Background}

Let $M$ be a nonnegative matrix. The \emph{nonnegative rank} of $M$, denoted $\rk^+(M)$, is the minimum $r$ such that $M$ can be decomposed as a sum $\sum_{i\in [r]} R_i$ where each $R_i$ is a rank-$1$ nonnegative matrix.

\emph{Randomized protocols.}
Faenza et al.~\cite{faenza14extended} observed that a nonnegative rank decomposition can be naturally interpreted as a type of randomized protocol that computes the matrix $M$ ``in expectation''. We phrase this connection precisely as follows: $\log\rk^+(M)+\Theta(1)$ is the minimum communication cost of a private-coin protocol $\Pi$ whose acceptance probability on each input $(x,y)$ satisfies $\Pr[\Pi(x,y)\text{ accepts}]=\alpha\cdot M_{x,y}$ where $\alpha>0$ is an absolute constant of proportionality (depending on $\Pi$ but not on $x,y$). All communication protocols in this paper are private-coin.

\emph{Slack matrices.}
The extension complexity of a polytope $P=\{x\in\R^n: Ax \geq b\}$ can be characterized in terms of the nonnegative rank of the \emph{slack matrix $M=M(P)$} associated with~$P$. The entries of $M$ are indexed by $(v,i)$ where $v\in P$ is a vertex of~$P$ and $i$ refers to the $i$-th facet-defining inequality $A_i x\geq b_i$ for $P$. We define $M_{v,i} \coloneqq A_i v-b_i\geq 0$ as the distance (\emph{slack}) of the $i$-th inequality from being tight for $v$. Yannakakis~\cite{yannakakis91expressing} showed that $\xc(P)=\rk^+(M(P))$.

A convenient fact for proving lower bounds on $\rk^+(M)$ is that the nonnegative rank is unaffected by the addition of columns to $M$ that each record the slack between vertices of $P$ and some valid (but not necessarily facet-defining) inequality for $P$. For notation, let $P\subseteq Q$ be two nested polytopes (in fact, $Q$ can be an unbounded polyhedron). We define $M(P;Q)$ as the slack matrix whose rows correspond to vertices of $P$ and columns correspond to the facets of $Q$ (hence $M(P;P)=M(P)$). We have $\rk^+(M(P))\ge\rk^+(M(P)\cup M(P;Q))-1\ge\rk^+(M(P;Q))-1$ where ``$\cup$'' denotes concatenation of columns.\footnote{Specifically, Farkas's Lemma implies that the slack of any valid inequality for $P$ can be written as a nonnegative linear combination of the slacks of the facet-defining inequalities for $P$, plus a nonnegative constant \cite[Proposition 1.9]{ziegler95lectures}. Thus if we take $M(P)\cup M(P;Q)$ and subtract off (possibly different) nonnegative constants from each of the ``new'' columns $M(P;Q)$, we get a matrix each of whose columns is a nonnegative linear combination of the ``original'' columns $M(P)$ and hence has the same nonnegative rank as $M(P)$. Since we subtracted off a nonnegative rank-$1$ matrix, we find that $\rk^+(M(P)\cup M(P;Q))\le\rk^+(M(P))+1$.} We summarize all the above in the following.
\begin{fact} \label{fact:xc-rk}
For all polytopes $P\subseteq Q$, we have $\xc(P) = \rk^+(M(P)) \geq \rk^+(M(P;Q))-1$.
\end{fact}

\section{\texorpdfstring{\KWEF}{KW/EF} Connection} \label{sec:kw-ef}

We now describe the connection showing that EF lower bounds follow from a certain type of strengthening of lower bounds for monotone KW games (and similarly, lower bounds for monotone KW games follow from certain strong enough EF lower bounds). This is not directly used in the proof of \autoref{thm:main}, but it serves as inspiration by suggesting the approach we use in the proof.

\subsection{Definitions} \label{sec:def-kw}
Let $f\colon\{0,1\}^n\to\{0,1\}$ be a monotone function. We define $\KWp(f)$ as the deterministic communication complexity of the following \emph{monotone KW game} associated with $f$.
\begin{myframe}{\boldmath$\KWp$-game}{0.4}
\begin{center}
\vspace{-1mm}
\setlength{\arraycolsep}{5pt}
\begin{tabular}{rl}
\emph{Input:}
& Alice gets $x\in f^{-1}(1)$, and Bob gets $y\in f^{-1}(0)$. \\
\emph{Output:}
& An index $i\in[n]$ such that $x_i=1$ and $y_i=0$.
\end{tabular}
\end{center}
\end{myframe}

We often think of $x$ and $y$ as subsets of $[n]$. In this language, a feasible solution for the $\KWp$-game is an $i\in x\cap \bar{y}$ where $\bar{y}\coloneqq [n]\smallsetminus y$. Given a monotone $f$, we denote by $F \coloneqq \conv f^{-1}(1)$ the associated polytope. We can express the fact that any pair $(x,y)\in f^{-1}(1)\times f^{-1}(0)$ admits at least one witness $i\in x\cap\bar{y}$ via the following linear inequality:
\begin{equation} \label{eq:kw-ineq}
\sum_{i\,:\,y_i=0} x_i~\geq~1.
\end{equation}
Since \eqref{eq:kw-ineq} is valid for all the vertices $x\in F$, it is valid for the whole polytope $F$. Define $F_{\KW} \supseteq F$ as the polyhedron whose facets are determined by the inequalities~\eqref{eq:kw-ineq}, as indexed by $0$-inputs~$y$. The $(x,y)$-th entry in the slack matrix $M(F;F_{\KW})$ is then $\sum_{i\,:\,y_i=0} x_i-1$. In words, this quantity counts the number of witnesses in the $\KWp$-game on input $(x,y)$ minus one.

More generally, let $S\subseteq \calX\times\calY\times\calQ$ be \emph{any} communication search problem (not necessarily a $\KWp$-game, even though any $S$ can be reformulated as such~\cite[Lemma 2.3]{gal01characterization}). Here $\calQ$ is some set of solutions/witnesses, and letting $S(x,y)\coloneqq\{q\in\calQ:(x,y,q)\in S\}$ denote the set of feasible solutions for input $(x,y)$, we assume that $S(x,y)\neq\emptyset$ for all $(x,y)$. We associate with $S$ the following natural \emph{``number of witnesses minus one''} communication game.
\begin{myframe}{\boldmath\witness-game}{0.37}
\begin{center}
\vspace{-1mm}
\setlength{\arraycolsep}{5pt}
\begin{tabular}{rl}
\emph{Input:}
& Alice gets $x\in \calX$, and Bob gets $y\in \calY$. \\
\emph{Output:}
& Accept with probability proportional to $|S(x,y)|-1$.
\end{tabular}
\end{center}
\end{myframe}
The communication complexity of this game is simply $\log\rk^+(M^S)+\Theta(1)$ where $M^S_{x,y}\coloneqq |S(x,y)|-1$.

\subsection{The connection}
What Hrube\v{s}~\cite[Proposition 4]{hrubes12nonnegative} observed was that an efficient protocol for a search problem~$S$ implies an efficient protocol for the associated \witness-game. In particular, for $\KWp$-games,
\begin{equation}\label{eq:kw-ef}
\log\rk^+(M(F;F_{\KW}))~\leq~O(\KWp(f)).
\tag{\KWEF}
\end{equation}
The private-coin protocol for $M(F;F_{\KW})$ computes as follows. On input $(x,y)\in f^{-1}(1)\times f^{-1}(0)$ we first run the optimal deterministic protocol for the $\KWp$-game for $f$ to find a particular $i\in[n]$ witnessing $x_i=1$ and $y_i=0$. Then, Alice uses her private coins to sample a $j\in[n]\smallsetminus\{i\}$ uniformly at random, and sends this $j$ to Bob. Finally, the two players check whether $x_j=1$ and $y_j=0$ accepting iff this is the case. The acceptance probability of this protocol is proportional to the number of witnesses minus one, and the protocol has cost $\KWp(f) + \log n + O(1)\leq O(\KWp(f))$ (where we assume w.l.o.g.\ that $f$ depends on all of its input bits so that $\KWp(f)\geq\log n$).

\subsection{Example: Matchings} \label{sec:matching}

\emph{Rothvo\ss~vs.{~\!}Raz--Wigderson.}
Consider the monotone function $f\colon\{0,1\}^{\binom{m}{2}}\to\{0,1\}$ that outputs~$1$ iff the input, interpreted as a graph on $m$ nodes ($m$ even), contains a perfect matching. Then $F\coloneqq \conv f^{-1}(1)$ is the perfect matching polytope. The inequalities \eqref{eq:kw-ineq} for $f$ happen to include the so-called ``odd set'' inequalities, which were exploited by Rothvo\ss~\cite{rothvos14matching} in showing that $\log\rk^+(M(F;F_{\KW}))\geq \Omega(m)$. Applying the \eqref{eq:kw-ef} connection to Rothvo\ss's lower bound implies in a black-box fashion that $\KWp(f)\geq \Omega(m)$, which is the result of Raz and Wigderson~\cite{raz92monotone}.

\bigskip\noindent\emph{Converse to \eqref{eq:kw-ef}?}
It is interesting to compare the above with the case of \emph{bipartite} perfect matchings. Consider a monotone $f\colon\{0,1\}^{m\times m}\to\{0,1\}$ that takes a bipartite graph as input and outputs $1$ iff the graph contains a perfect matching. It is well-known that $F\coloneqq\conv f^{-1}(1)$ admits a polynomial-size extended formulation~\cite[Theorem 18.1]{schrijver03combinatorial}. By contrast, the lower bound $\KWp(f)\geq\Omega(m)$ from~\cite{raz92monotone} continues to hold even in the bipartite case. This example shows that the converse inequality to \eqref{eq:kw-ef} does not hold in general. Hence, a lower bound for the \witness-game can be a strictly stronger result than a similar lower bound for the $\KWp$-game.

\subsection{Minterms and maxterms} \label{sec:minterm}

A \emph{minterm} $x\in f^{-1}(1)$ is a minimal $1$-input in the sense that flipping any $1$-entry of $x$ into a $0$ will result in a $0$-input. Analogously, a \emph{maxterm} $y\in f^{-1}(0)$ is a maximal $0$-input. It is a basic fact that solving the $\KWp$-game for minterms/maxterms is enough to solve the search problem on any input: Say that Alice's input $x$ is not a minterm. Then Alice can replace $x$ with any minterm $x'\subseteq x$ and run the protocol on $x'$. A witness $i\in[n]$ for $(x',y)$ works also for $(x,y)$. A similar fact holds for the \witness-game: we claim that the nonnegative rank does not change by much when restricted to minterms/maxterms. Say that Alice's input $x$ is not a minterm. Then Alice can write $x=x'\cup x''$ (disjoint union) where $x'$ is a minterm. Then $|x\cap \bar{y}| -1 = (|x'\cap\bar{y}|-1) + |x''\cap \bar{y}|$ where the first term is the \witness-game for $(x',y)$ and the second term has nonnegative rank at most~$n$. (A similar argument works if Bob does not have a maxterm.)

\section{Tseitin Problem} \label{sec:tseitin-def}

\subsection{Query version}

Fix a connected node-labeled graph $G=(V,E,\ell)$ where $\ell\in \Z_2^V$ has \emph{odd weight}, i.e., $\sum_{v\in V}\ell(v)=1$ where the addition is modulo $2$. For any edge-labeling $z\in \Z_2^E$ and a node $v\in V$ we write concisely $z(v)\coloneqq \sum_{e\ni v} z(e)$ for the mod-$2$ sum of the edge-labels adjacent to $v$.
\begin{myframe}{Tseitin problem:~\normalfont $\TSE_G$}{0.3}
\begin{center}
\vspace{-4mm}
\setlength{\arraycolsep}{5pt}
\begin{tabular}{rl}
\emph{Input:}
& Labeling $z\in \Z_2^E$ of the edges. \\[0mm]
\emph{Output:}
& A node $v\in V$ containing a \emph{parity violation} $z(v) \neq \ell(v)$.
\end{tabular}
\end{center}
\end{myframe}
As a sanity check, we note that on each input~$z$ there must exist at least one node with a parity violation. This follows from the fact that, since each edge has two endpoints, the sum $\sum_v z(v)$ is even, whereas we assumed that the sum $\sum_v \ell(v)$ is odd.

\paragraph{Basic properties.}
The above argument implies more generally that the set of violations $\viol(z)\coloneqq \{ v\in V: z(v)\neq \ell(v)\}$ is always of odd size. Conversely, for any odd-size set $S\subseteq V$ we can design an input~$z$ such that $\viol(z)=S$. To see this, it is useful to understand what happens when we \emph{flip a path} in an input $z$. Formally, suppose $p\in \Z_2^E$ is (an indicator vector of) a path. Define $z^p$ as $z$ with bits on the path $p$ flipped (note that $z^p= z+p\in\Z_2^E$; however, the notation $z^p$ will be more convenient later). Flipping $p$ has the effect of flipping whether each endpoint of $p$ is a violation. More precisely, the violated nodes in $z^p$ are related to those in $z$ as follows: (i) if both endpoints of $p$ are violated in $z$ then the flip causes that pair of violations to disappear; (ii) if neither endpoint of $p$ is violated in $z$, then the flip introduces a pair of new violations; (iii) if precisely one endpoint of $p$ was violated in $z$, then the flip moves a violation from one endpoint of $p$ to the other. By applying (i)--(iii) repeatedly in a connected graph $G$, we can design an input~$z$ where $\viol(z)$ equals any prescribed odd-size set $S$.

If $z$ and $z'$ have the same set of violations, $\viol(z)=\viol(z')$, then their difference $q\coloneqq z-z'\in\Z_2^E$ satisfies $q(v)=0$ for all $v\in V$. That is, $q$ is an \emph{eulerian} subgraph of $G$. On the other hand, for any eulerian graph $q$, the inputs $z$ and $z^q$ have the same violations. Consequently, to generate a random input with the same set of violations as some fixed $z$, we need only pick a random eulerian graph $q$ and output $z^q$. (Eulerian graphs form a subspace of $\Z_2^E$, sometimes called the \emph{cycle space} of $G$.)

\subsection{Communication version}

The communication version of the Tseitin problem is obtained by composing (or \emph{lifting}) $\TSE_G$ with a constant-size two-party gadget $g\colon\calX\times\calY\to\{0,1\}$. In the lifted problem $\TSE_G\circ g^n$, where $n\coloneqq |E|$, Alice gets $x\in\calX^n$ as input, Bob gets $y\in\calY^n$ as input, and their goal is to find a node $v\in V$ that is violated for
\[
z~\coloneqq~g^n(x,y)~=~(g(x_1,y_1),\ldots,g(x_n,y_n)).
\]
We define our gadget precisely in \autoref{sec:gadget}. For now---in particular, for the reductions presented in the next section---the only important property of our gadget is that $|\calX|,|\calY|\leq O(1)$.

\subsection{Statement of result}

We prove that there is a family of bounded-degree graphs $G$ such that the \witness-game associated with $\TSE_G\circ g^n$ requires $\Omega(n/\log n)$ bits of communication. We prove our lower bound assuming only that $G=(V,E)$ is well-connected enough as captured by the following definition (also used in~\cite{goos14communication}). A graph $G$ is \emph{$k$-routable} iff there is a set of $2k+1$ nodes $T\subseteq V$ called \emph{terminals} such that for any \emph{pairing} $\calP\coloneqq \{\{s_i,t_i\}:i\in [\kappa]\}$ (set of pairwise disjoint pairs) of $2\kappa$ terminals ($\kappa\le k$), there exist $\kappa$ edge-disjoint paths (called \emph{canonical} paths for $\calP$) such that the $i$-th path connects $s_i$ to $t_i$. Furthermore, we tacitly equip $G$ with an arbitrary odd-weight node-labeling.

\begin{theorem} \label{thm:communication-lb}
There is a constant-size $g$ such that for every $k$-routable graph $G$ with $n$ edges, the \witness-game for $\TSE_G\circ g^n$ requires $\Omega(k)$ bits of communication.
\end{theorem}

If we choose $G$ to be a sufficiently strong expander graph, we may take $k=\Theta(n/\log n)$ as shown by Frieze et al.~\cite{frieze00optimal,frieze01edge}. Alternative constructions with $k=\Theta(n/\log n)$ exist based on bounded-degree ``butterfly'' graphs; see~\cite[\S5]{nordstroem15new} for an exposition.

\begin{corollary} \label{cor:expander}
There is a constant-size $g$ and an explicit bounded-degree graph $G$ with $n$ edges such that the \witness-game for $\TSE_G\circ g^n$ requires $\Omega(n/\log n)$ bits of communication.
\end{corollary}

As a bonus, we also prove that the \emph{query complexity} of the \witness-game for $\TSE_G$ is $\Omega(n)$ on any expander $G$ (see \autoref{sec:query-lb}).

\section{Reductions} \label{sec:reductions}

The goal of this section is to show, via reductions, that a lower bound on the \witness-game for $\TSE_G\circ g^n$ (where $G=(V,E)$ is of bounded degree and $n\coloneqq|E|$) translates directly into a lower bound on the extension complexity of $P_K$ for an $O(n)$-node bounded-degree graph $K$.

\subsection{Definition: Monotone CSP-SAT}

We start by describing a way of representing constraint satisfaction problems (CSP) as a monotone function; this was introduced in~\cite{goos14communication} and further studied by Oliveira~\cite[Chapter~3]{oliveira15unconditional}. The function is defined relative to some finite alphabet $\Sigma$ and a fixed constraint topology determined by a bipartite graph~$H\coloneqq(L\cup R,E)$. The left nodes $L$ are thought of as \emph{variables} (taking values in $\Sigma$) and the right nodes $R$ correspond to \emph{constraints}. For a constraint $c\in R$, let $\var(c)\subseteq L$ denote the variables involved in $c$. Let $d$ denote the maximum degree of a node in $R$. The function $\SAT=\SAT_{\Sigma,H}\colon \{0,1\}^m\to\{0,1\}$, where $m\leq |R|\cdot|\Sigma|^d$, is now defined as follows. An input $x\in\{0,1\}^m$ defines a CSP instance by specifying, for each $c\in R$, a truth table $\Sigma^{\var(c)}\to\{0,1\}$ that records which assignments to the variables $\var(c)$ satisfy $c$. Then $\SAT(x)\coloneqq 1$ iff there is some global assignment $L\to\Sigma$ that satisfies all the constraints as specified by $x$. This is monotone: if we flip any $0$ into a $1$ in the truth table of a constraint, we are only making the constraint easier to satisfy.

\subsection{From Tseitin to CSP-SAT} \label{sec:tse-to-sat}

For completeness, we present the reduction (due to \cite[\S5.1]{goos14communication}) from the search problem $\TSE_G\circ g^n$ to the $\KWp$-game for $\SAT=\SAT_{\calX,H}\colon\{0,1\}^m\to\{0,1\}$. Here the alphabet is $\calX$ and the bipartite graph $H$ is defined on $E(G)\cup V(G)$ such that there is an edge $(e,v)\in E(H)$ iff $v\in e$. Note that $m\leq O(n)$ provided that $|\calX|\leq O(1)$ and that $G$ is of bounded degree.

On input $(x,y)$ to $\TSE_G\circ g^n$ the two players proceed as follows:
\begin{itemize}[itemsep=2mm]
\item Alice maps her $x\in\calX^{E(G)}$ into a CSP whose sole satisfying assignment is $x$. Namely, for each constraint $v\in V(G)$, the truth table $\calX^{\var(v)}\to\{0,1\}$ is all-0 except for a unique 1 in position~$x|_{\var(v)}$ (restriction of $x$ to coordinates in $\var(v)$).
\item Bob maps his $y\in\calY^{E(G)}$ into an unsatisfiable CSP. Namely, for each constraint $v\in V(G)$, the truth table $t_v\colon\calX^{\var(v)}\to\{0,1\}$ is given by $t_v(\hat{x})\coloneqq 1$ iff $(g(\hat{x}_e,y_e))_{e\in\var(v)}\in\{0,1\}^{\var(v)}$ is a partial edge-labeling of $G$ that does \emph{not} create a parity violation on $v$.
\end{itemize}
Let us explain why Bob really produces a $0$-input of $\SAT$. Suppose for contradiction that there is an~$\hat{x}\in\calX^{E(G)}$ that satisfies all of Bob's constraints: $t_v(\hat{x}|_{\var(v)})=1$ for all $v$. By definition, this means that $z\coloneqq g^n(\hat{x},y)$ is an input to $\TSE_G$ without any violated nodes---a contradiction.

This reduction is parsimonious: it maps witnesses to witnesses in 1-to-1 fashion. Indeed, a node~$v$ is violated for $\TSE_G\circ g^n$ if and only if Alice's truth table for $v$ has its unique $1$ in a coordinate where Bob has a $0$. In conclusion, the \witness-game associated with (the $\KWp$-game for) $\SAT$ is at least as hard as the \witness-game for $\TSE_G\circ g^n$.

\subsection{From CSP-SAT to independent sets} \label{sec:sat-to-is}

As a final step, we start with $\SAT=\SAT_{\Sigma,H}\colon\{0,1\}^m\to\{0,1\}$ and construct an $m$-node graph $K$ such that a slack matrix of the independent set polytope $P_K$ embeds the \witness-game for $\SAT$ (restricted to minterms). Let $H\coloneqq(L\cup R,E)$ (as above) and define $n\coloneqq |R|$ (above we had $n=|L|$, but in our case $|L|=\Theta(|R|)$ anyway).

The $m$-node graph $K$ is defined as follows (this is reminiscent of a reduction from~\cite{feige96interactive}).
\begin{itemize}[itemsep=2mm]
\item The nodes of $K$ are in 1-to-1 correspondence with the input bits of $\SAT$. That is, for each constraint $c\in R$ we have $|\Sigma^{\var(c)}|$ many nodes in $K$ labeled with assignments $\var(c)\to\Sigma$.
\item There is an edge between any two nodes whose assignments are \emph{inconsistent} with one another. (Here $\phi_i\colon\var(c_i)\to\Sigma$, $i\in\{1,2\}$, are inconsistent iff there is some $e\in\var(c_1)\cap\var(c_2)$ such that $\phi_1(e)\neq \phi_2(e)$.) In particular, the truth table of each constraint becomes a clique.
\end{itemize}
(It can be seen that $K$ has bounded degree if $H$ has bounded left- and right-degree, which it does after our reduction from Tseitin for a bounded-degree $G$.)

The key property of this construction is the following:
\begin{center}
\itshape The minterms of $\SAT$ are precisely the (indicator vectors of) maximal independent sets of $K$.
\end{center}
Indeed, the minterms $x\in \SAT^{-1}(1)$ correspond to CSPs with a unique satisfying assignment $\phi\colon L\to\Sigma$; there is a single 1-entry in each of the $n$ truth tables (so that $|x|=n$) consistent with $\phi$. Such an $x$, interpreted as a subset of nodes, is independent in~$K$ as it only contains nodes whose labels are consistent with $\phi$. Conversely, because every independent set $x\subseteq V(K)$ can only contain pairwise consistently labeled nodes, $x$ naturally defines a partial assignment $L'\to\Sigma$ for some $L'\subseteq L$. A maximal independent set $x$ corresponds to picking a node from each of the $n$ constraint cliques consistent with some total assignment $\phi\colon L\to\Sigma$. Hence $x$ is a $1$-input to $\SAT$ with unique satisfying assignment $\phi$.

Our goal is now to exhibit a set of valid inequalities for the independent set polytope $P_K$ whose associated slack matrix embeds the \witness-game for $\SAT$. Let $x\subseteq V(K)$ be an independent set and $y\in\SAT^{-1}(0)$. We claim that the following inequalities (indexed by $y$) are valid:
\begin{equation} \label{eq:minus-one}
|x\cap y|~=~\sum_{i\,:\,y_i=1}x_i~\leq~n-1.
\end{equation}
Clearly \eqref{eq:minus-one} holds whenever $|x|\leq n-1$. Since it is impossible to have $|x|\geq n+1$, assume that $x$ is maximal:~$|x|=n$. As argued above, $x$ is a minterm of $\SAT$. Hence $(x,y)$ is a valid pair of inputs to the $\KWp$-game, and so they admit a witness: $|x\cap\bar{y}|\geq 1$. Therefore $|x\cap y| = n - |x\cap \bar{y}|\leq n-1$. This shows that \eqref{eq:minus-one} is valid. The slack matrix associated with inequalities \eqref{eq:minus-one} has entries
\[
n-1-|x\cap y|~=~|x\cap \bar{y}|-1,
\]
for any minterm $x$ and any $y\in\SAT^{-1}(0)$. But this is just the \witness-game for $\SAT$ with Alice's input restricted to minterms.

\subsection[Proof of \autoref*{thm:main}]{Proof of \autoref{thm:main}}

Here we simply string the above reductions together. By \autoref{cor:expander} there is a constant-size~$g$ and a bounded-degree $G$ with $n$ edges such that the \witness-game for $\TSE_G\circ g^n$ requires $\Omega(n/\log n)$ bits of communication. By the reduction of \autoref{sec:tse-to-sat} this implies an $\Omega(n/\log n)$ lower bound for the \witness-game associated with (the $\KWp$-game for) a monotone function $\SAT\colon\{0,1\}^{O(n)}\to\{0,1\}$. As discussed in \autoref{sec:minterm}, the complexity of the \witness-game for $\SAT$ is affected only by $\pm \log n$ when restricted to minterms. Thus the minterm-restricted \witness-game for $\SAT$ still has complexity $\Omega(n/\log n)$. (Alternatively, one can note that the reduction from Tseitin to CSP-SAT produced only minterms.) Hence the nonnegative rank of the matrix for that game is $2^{\Omega(n/\log n)}$. By the reduction of \autoref{sec:sat-to-is} there is a bounded-degree $O(n)$-node graph $K$ and a system of valid inequalities~\eqref{eq:minus-one} for the independent set polytope $P_K$ such that the slack matrix $M(P_K;Q)$, where $Q$ is the polyhedron with facets determined by~\eqref{eq:minus-one}, embeds the matrix for the minterm-restricted \witness-game for $\SAT$. Thus $\log\rk^+(M(P_K;Q))\geq \Omega(n/\log n)$. By \autoref{fact:xc-rk} we have $\log\xc(P_K)=\log\rk^+(M(P_K))\geq\log\bigl(\rk^+(M(P_K;Q))-1\bigr)\geq\Omega(n/\log n)$.

\section{Our Gadget} \label{sec:gadget}

We define our two-party gadget $g\colon\{0,1\}^3\times\{0,1\}^3\to\{0,1\}$ as follows; see \autoref{fig:gadget}:
\[
g(x,y)~\coloneqq~ x_1 + y_1 + x_2y_2 + x_3y_3\pmod{2}.
\]
We note that the smaller gadget $x_1+y_1+x_2y_2\pmod{2}$ was considered in~\cite{sherstov11pattern,goos14communication}.

\subsection{Flips and windows} \label{sec:windows}
The most basic property of $g$ is that it admits \emph{Alice/Bob-flips}:
\begin{enumerate}[label=(\arabic*),itemsep=2mm]
\item \emph{Alice-flips:} There is a row permutation $\pi_\A\colon\calX\to\calX$ that flips the output of the gadget: $g(\pi_\A(x),y)=\neg g(x,y)$ for all $x,y$. Namely, Alice just flips the value of $x_1$.
\item \emph{Bob-flips:} There is a column permutation $\pi_\B\colon\calY\to\calY$ that flips the output of the gadget: $g(x,\pi_\B(y))=\neg g(x,y)$ for all $x,y$. Namely, Bob just flips the value of $y_1$.
\end{enumerate}

\begin{figure}
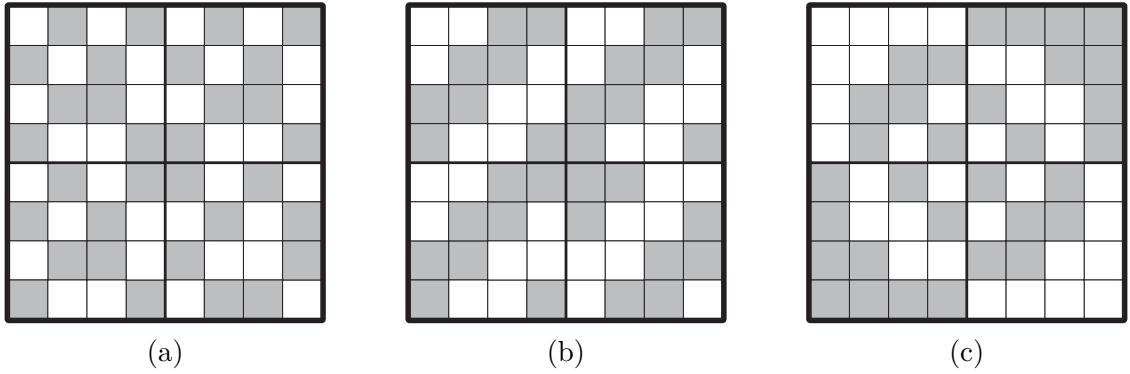
%
\centering
\begin{lpic}[t(-5mm)]{gadget1(.26)}
\lbl[c]{100,2;(a)}
\end{lpic}
\begin{lpic}[t(-5mm)]{gadget2(.26)}
\lbl[c]{100,2;(b)}
\end{lpic}
\begin{lpic}[t(-5mm)]{gadget3(.26)}
\lbl[c]{100,2;(c)}
\end{lpic}\\[2mm]
\caption{Three ways to view our gadget $g\colon\calX\times\calY\to\{0,1\}$ by permuting rows and columns. The white and gray cells represent $0$- and $1$-inputs, respectively.}
\label{fig:gadget}
\vspace{2mm}
\end{figure}

A more interesting feature of our gadget (which $x_1+y_1+x_2y_2$ does not possess) is that $g$ embeds---in an especially uniform manner---certain $2\times 4$ and $4\times 2$ submatrices which we call ``stretched $\AND$'' and ``stretched $\NAND$''. For terminology, we define a \emph{$z$-window} where $z\in\{0,1\}$ as a $z$-monochromatic rectangle of size $2$ in the domain of $g$, i.e., an all-$z$ submatrix of either horizontal shape $1\times 2$ or vertical shape $2\times 1$. Here is an illustration of \emph{horizontally} stretched $\AND/\NAND$, which are composed of four horizontally shaped windows (for \emph{vertical} stretch, the illustration should be transposed):
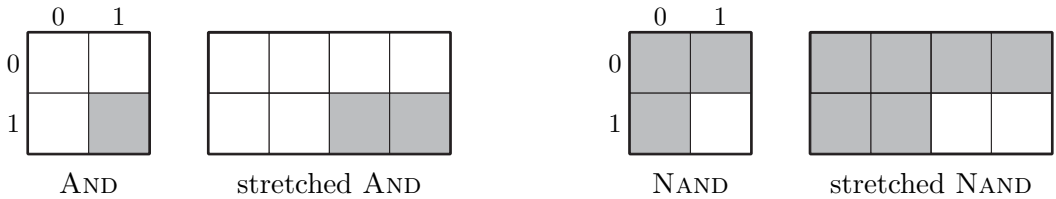
\begin{center}
\begin{lpic}[b(3mm)]{stretched(.2)}
\lbl[c]{60,0;$\AND$}
\lbl[c]{220,0;stretched $\AND$}
\lbl[c]{460,0;$\NAND$}
\lbl[c]{620,0;stretched $\NAND$}

\small
\lbl[c]{10,40;$1$}
\lbl[c]{10,80;$0$}
\lbl[c]{40,111;$0$}
\lbl[c]{80,111;$1$}

\lbl[c]{410,40;$1$}
\lbl[c]{410,80;$0$}
\lbl[c]{440,111;$0$}
\lbl[c]{480,111;$1$}

\end{lpic}
\end{center}
The key property is that each $z$-window $w$ is embedded as the stretched $(1,1)$-input to a \emph{unique} embedding of stretched $\AND$ (if $z=1$) or $\NAND$ (if $z=0$) inside $g$. That is, for each $w$ we can find the following unique submatrix (illustrated again for horizontal shapes), where we denote by $w^\leftarrow$, $w^\nwarrow$, and $w^\uparrow$ the $(1-z)$-windows corresponding to the stretched $(1,0)$-, $(0,0)$-, and $(0,1)$-inputs to the stretched $\AND/\NAND$.
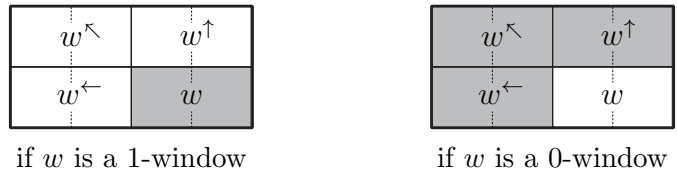
\begin{center}
\begin{lpic}[b(3mm)]{flips(.2)}

\lbl[c]{100,0; if $w$ is a $1$-window}
\lbl[c]{380,0; if $w$ is a $0$-window}
\large
\lbl[c]{66,83;$w^\nwarrow$}
\lbl[c]{66,42;$w^\leftarrow$}
\lbl[c]{144,83;$w^\uparrow$}
\lbl[c]{140,40;$w$}

\lbl[c]{346,83;$w^\nwarrow$}
\lbl[c]{346,42;$w^\leftarrow$}
\lbl[c]{424,83;$w^\uparrow$}
\lbl[c]{420,40;$w$}

\end{lpic}
\end{center}
This defines three maps (``\emph{directed flips}'') $w\mapsto w^\leftarrow$, $w\mapsto w^\nwarrow$, $w\mapsto w^\uparrow$, which turn out to be shape-maintaining \emph{bijections} between the set of $z$-windows and the set of $(1-z)$-windows. In particular, if $w$ is a uniformly random $z$-window of~$g$, then each of $w^\leftarrow$, $w^\nwarrow$, $w^\uparrow$ is a uniformly random $(1-z)$-window.

\subsection{Checking the existence of flips}

The properties of $g$ claimed above can be verified by directly inspecting the gadget (by hand). Luckily, this task can be eased by exploiting symmetries.
\begin{enumerate}[label=(\arabic*),itemsep=2mm,resume]
\item \emph{Transitive symmetry}: The gadget admits a group of symmetries (permutations of its rows and columns leaving $g$ invariant) which splits the domain of~$g$ into two orbits, $g^{-1}(1)$ and $g^{-1}(0)$. Specifically, there is a group $\calS\subseteq \mathfrak{S}_8\times\mathfrak{S}_8$ (here $\mathfrak{S}_8$ is the symmetric group on $8$ elements) such that when $(\pi_1,\pi_2)\in \calS$ acts on $g$, the output remains invariant: $g(\pi_1(x),\pi_2(y))=g(x,y)$ for all $x,y$; and moreover, $\calS$ is transitive in the sense that for any two $1$-inputs $(x,y),(x',y')\in g^{-1}(1)$ (or $0$-inputs) there is a symmetry $(\pi_1,\pi_2)\in \calS$ such that $(\pi_1(x),\pi_2(y))=(x',y')$.
\end{enumerate}
To see that $g$ really does have property (3), we visualize $g$ as constructed from $\XOR(x_1,x_2)\coloneqq x_1+x_2\pmod{2}$ by applying the following~``$\leadsto$'' transformation twice:
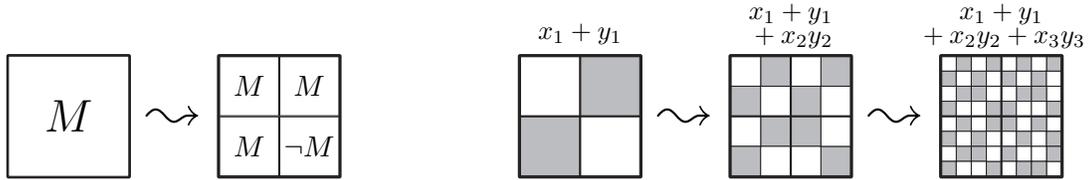
\begin{center}
\begin{lpic}[b(-2mm),t(0mm)]{construction(.2)}
\lbl[c]{60,60;\LARGE$M$}

\lbl[c]{180,40;$M$}
\lbl[c]{180,80;$M$}
\lbl[c]{220,40;$\neg M$}
\lbl[c]{220,80;$M$}

\small
\lbl[c]{400,115;$x_1+y_1$}
\lbl[c]{540,128;$x_1+y_1$}
\lbl[c]{540,112;$\mbox{}+x_2y_2$}
\lbl[c]{680,128;$x_1+y_1$}
\lbl[c]{680,112;$\mbox{}+x_2y_2+x_3y_3$}

\huge
\lbl[c]{130,58;$\leadsto$}
\lbl[c]{470,58;$\leadsto$}
\lbl[c]{610,58;$\leadsto$}
\end{lpic}
\end{center}
It is easy to see that $\XOR$ has the properties (1)--(3). We argue that if $M$ is a boolean matrix with the properties (1)--(3) and $M\leadsto M'$, then $M'$ has the properties (1)--(3). Suppose the entries of $M$ are indexed by $(x,y)$; we use $(xa,yb)$ to index the entries of $M'$ where $a,b\in\{0,1\}$ are bits. If $\pi_\A$, $\pi_\B$ are the Alice/Bob-flips for~$M$, then Alice/Bob-flips for $M'$ are
\begin{align*}
xa~&\mapsto~\pi_\A(x)a,\\
yb~&\mapsto~\pi_\B(y)b.
\intertext{Suppose $\calS$ is the transitive symmetry group for $M$. Then the transitive symmetry group for $M'$ is generated by the following symmetries (here $\pi_\A^0(x)\coloneqq x$ and $\pi_\A^1(x)\coloneqq\pi_\A(x)$ and similarly for $\pi_\B^b$):}
\forall (\pi_1,\pi_2)\in\calS:\quad (xa,yb)~&\mapsto~(\pi_1(x)a,\pi_2(y)b),\hspace{1.8cm}\\
(xa,yb)~&\mapsto~(\pi_\A^a(x)a,y(1-b)),\\
(xa,yb)~&\mapsto~(x(1-a),\pi_\B^b(y)b).
\end{align*}
The first family of symmetries makes each quadrant of $M'$ transitive, whereas the last two symmetries map entries between quadrants. In the second-to-last symmetry, Bob swaps the left and right halves while Alice applies her flip to the bottom half. In the last symmetry, Alice swaps the top and bottom halves while Bob applies his flip to the right half. This shows that $g$ satisfies (1)--(3).

Rather than checking that each $z$-window $w$ appears as the stretched $(1,1)$-input to a unique embedding of stretched $\AND/\NAND$ and that the directed flips are bijections, it is equivalent to check that for all $\ell\in\{(0,0),(0,1),(1,0),(1,1)\}$ each $w$ appears as the stretched $\ell$-input to a unique embedding of stretched $\AND/\NAND$ in $g$. Let us check this assuming $w$ is a $0$-window of shape $1\times 2$ (the other possibilities can be checked similarly). By transitive symmetry, we may assume that $w$ is picked among the four 0's of the first row of \autoref{fig:gadget}(c) (so $\binom{4}{2}$ choices for $w$). The key observation is that the four columns corresponding to these 0's define a submatrix of $g$ (left half of (c)) that contains each even Hamming weight row once, and that the other four columns (right half of (c)) also contain each even Hamming weight row once. We consider the four cases for $\ell$.
\begin{itemize}[leftmargin=2cm]
\item[$\ell=(0,0)$:] To see that $w$ is the stretched $(0,0)$-input to a unique embedding of stretched $\AND$, find the unique other row that has 0's in the same columns as $w$. The other two columns in the left half of (c) have 0's in the top row and 1's in the other row.
\item[$\ell=(0,1)$:] To see that $w$ is the stretched $(0,1)$-input to a unique embedding of stretched $\AND$, find the unique other row that has 1's in the same columns as $w$ and 0's in the other two columns of the left half of (c). These other two columns have 0's in the top row.
\item[$\ell=(1,0)$:] To see that $w$ is the stretched $(1,0)$-input to a unique embedding of stretched $\AND$, find the unique other row that has 0's in the same columns as $w$, then find the unique pair of columns in the right half of (c) that has 0's in that other row. This pair of columns has 1's in the first row.
\item[$\ell=(1,1)$:] To see that $w$ is the stretched $(1,1)$-input to a unique embedding of stretched $\NAND$, find the unique other row that has 1's in the same columns as $w$ and 0's in the other two columns of the left half of (c), then find the unique pair of columns in the right half of (c) that has 1's in that other row. This pair of columns has 1's in the first row.
\end{itemize}

\section{Communication Lower Bound} \label{sec:communication-lb}

In this section we prove \autoref{thm:communication-lb}, where $g$ is the gadget from \autoref{sec:gadget}.

\subsection{High-level intuition}

The high-level reason for why the \witness-game for Tseitin (or really for any sufficiently unstructured search problem) is hard is the same as for the \witness-game for matching~\cite{rothvos14matching}: A correct protocol~$\Pi$ dare not accept its input before it has found at least two witnesses, lest it risk accepting with positive probability an input with a unique witness (which would contradict correctness). However, in an input with $i$ witnesses, there are $\binom{i}{2}$ pairs of witnesses for the protocol to find. Hence one expects the acceptance probability of $\Pi$ (that communicates too few bits and never errs when $i=1$) to grow at least \emph{quadratically} with $i$ rather than linearly as required by the \witness-game.

Formalizing this quadratic increase in acceptance probability for protocols takes some technical work given the current tools available in communication complexity. However, the quadratic increase phenomenon for Tseitin is easier to formalize in the query complexity setting, which we do in \autoref{sec:query-lb}. The reader may want to have a look at that simpler proof first, even though the query proof is somewhat incomparable to our approach for protocols (which revolves around $k$-routability).

\subsection{Preliminaries} \label{sec:preliminaries}

\paragraph{Probability and information theory.}
We use standard notions from information theory: $\H(X)$ is Shannon entropy; $\H(X\mid Y)\coloneqq \E_{y\sim Y}\H(X\mid Y=y)$ is conditional entropy; $\I(X\semi Y)\coloneqq\H(X)-\H(X\mid Y)=\H(Y)-\H(Y\mid X)$ is mutual information; $\Delta(X,Y)$ is statistical (total variation) distance. We use upper-case letters for random variables and corresponding lower-case letters for particular outcomes. Throughout the whole proof, all random choices are assumed to be uniform in their respective domains unless otherwise stated.

\paragraph{Inputs and transcripts.}
Let $XY$ be random inputs to a private-coin protocol~$\Pi$. We denote by $\Pi=\Pi(X,Y)$ the transcript of the protocol on input $XY$, and we let $|\Pi|$ be the maximum length of a transcript (i.e., the communication cost of $\Pi$). Note that the transcript $\Pi$ depends on both $XY$ and the private coins of the players. We let $\Pia \coloneqq (\Pi\mid \Pi\text{ accepts})$ denote the transcript conditioned on the protocol accepting. For each input $z\in\Z_2^n$ to the query problem $\TSE_G$ we can associate in a natural way a pair of random inputs $XY$ to the communication problem $\TSE_G\circ g^n$ that are \emph{consistent with~$z$} in the sense that $g^n(X,Y)=z$; namely, we let $XY$ be uniformly distributed on
\[
(g^n)^{-1}(z)~=~g^{-1}(z_1)\times\cdots\times g^{-1}(z_n).
\]
We write $\Pi|z$ as a shorthand for $\Pi(X,Y)$ where $XY$ are drawn at random from the above set.

\paragraph{Windows.}
As is often the case with information complexity arguments, we need to introduce a conditioning variable $W$ whose purpose is to make $X$ and $Y$ conditionally independent. To this end, we employ windows (\autoref{sec:windows}): we call a rectangle $w\coloneqq w_1\times\cdots\times w_n\subseteq (g^n)^{-1}(z)$ a (multi-gadget) \emph{window of~$z$} iff each $w_i$ is a $z_i$-window in $g$ (so $w_i\subseteq g^{-1}(z_i)$). Now, to generate $XY$ as above, we first pick $W$ uniformly at random among all the windows of $z$, and then, conditioned on an outcome $W=w$, we pick $XY\in w$ uniformly at random. In conclusion, $XY$ is uniform on $(g^n)^{-1}(z)$ (since each row and column of $g$ is balanced) and $X$ and $Y$ are conditionally independent given $W$. We write $\Pi|w\coloneqq (\Pi(X,Y)\mid W=w)$ for short.

\paragraph{Alice-flips.}
Let $(x,y)$ be an input consistent with $z\coloneqq g^n(x,y)$ and let $B\subseteq [n]$ be any subset of coordinates of $z$. ($B$ stands for ``block'' by analogy with the concept of block sensitivity from query complexity.) We denote by $(x^B,y)$ the input obtained from $(x,y)$ by letting Alice flip the outputs of all gadgets corresponding to coordinates in $B$, i.e., for every $i\in B$ Alice replaces her input $x_i$ with $\pi_\A(x_i)$ where $\pi_\A$ is the row permutation from~\autoref{sec:windows}. Hence $(x^B,y)$ is an input consistent with $z^B$. We can also have Alice flip whole windows: $w^B\coloneqq\{(x^B,y):(x,y)\in w\}$. We henceforth refer to such Alice-flips as just ``flips''. (We could equally well have Bob be the flipper throughout the whole proof, but we needed to make an arbitrary choice between the players.)

\paragraph{Smooth protocols.}
Recall that if $z$ is an input to $\TSE_G$ and $B\subseteq E(G)$ is an eulerian graph, then $z$ and $z^B$ have the same set of violations. Consequently, any protocol $\Pi$ for the \witness-game must accept inputs $(x,y)$ and $(x^B,y)$ with the same probability. We note that we may assume w.l.o.g.\ that the transcript distribution of $\Pi$ is not sensitive to flipping eulerian graphs: if $w$ is a window and $B$ an eulerian graph, then $\Pi|w$ and $\Pi|w^B$ have the same distribution. Indeed, if $\Pi$ does not satisfy this, then we may replace it by a new ``smoothed'' protocol $\Pi'$ that computes as follows on input $(x,y)$: Alice uses her private coins to choose a uniformly random eulerian graph $B$ and then the players run $\Pi$ on input $(x^B,y)$. The fact that we may assume $\Pi$ is smooth is critically used later in the proof.

\subsection{Proof outline}

Let us assume for the sake of contradiction that $\Pi$ is a private-coin protocol of cost~$|\Pi|\leq o(k)$ that accepts each input $(x,y)$ with probability $\alpha\cdot(|\!\viol(z)|-1)$ where $\alpha>0$ is a constant (independent of $(x,y)$) and $z\coloneqq g^n(x,y)$. We call an input $z$ (and any $(x,y)$ consistent with $z$) an \emph{$i$-violation input} if $|\!\viol(z)|=i$ and all violations occur at the terminals $T$. We analyze the behavior of $\Pi$ on $i$-violation inputs with $i \in\{1,3,7\}$ and show a contradiction via the following implication:
\begin{enumerate}[label=(\boldmath$*$)]
\item \label{implication} \itshape
If protocol $\Pi$ accepts all 1-violation (resp.\ 3-violation) inputs with probability $0$ (resp.~$2\alpha$),\\
then $\Pi$ must mess up by accepting some 7-violation input with probability $> 6\alpha$.
\end{enumerate}
Henceforth, we use $o(1)$ to denote anonymous quantities that tend to $0$ as $|\Pi|/k$ tends to $0$.

The implication \ref{implication} can be derived cleanly from two types of limitations of our too-good-to-be-true $\Pi$. The first limitation concerns the situation where we start with a $1$-violation input $z$, and consider $3$-violation inputs $z^{B_1}$ and $z^{B_2}$ that are obtained from $z$ by flipping either a typical path $B_1$ or another typical path $B_2$ that is edge-disjoint from $B_1$ (the endpoints of $B_i$ are terminals). The protocol should accept both $z^{B_1}$ and $z^{B_2}$ (more precisely, any $(x,y)$ consistent with them) with probability~$2\alpha$, but it better not accept both inputs while generating the same transcript---otherwise we could cut-and-paste $z^{B_1}$ and $z^{B_2}$ together and fool $\Pi$ into accepting $z$ (which would contradict correctness). What we actually get is that the accepting transcripts for $z^{B_1}$ and $z^{B_2}$ should be near-disjoint:
\begin{one-vs-three} \label{thm:1v3}
Let $z$ be any 1-violation input and let $\calP$ be any pairing of the non-violated terminals with canonical edge-disjoint paths $B_1,\ldots,B_k$. Let $w$ be a random window of~$z$, and choose distinct $i,j\in[k]$ at random. Then, with probability $\geq 1-o(1)$,
\[
\Delta\bigl(\Pia|w^{B_i},\Pia|w^{B_j}\bigr)~\geq~1-o(1).
\]
\end{one-vs-three}

The second limitation concerns the situation where we start with a 3-violation input $z$ and flip a typical path $B$ to obtain a 5-violation input $z^B$. Consider a typical accepting transcript $\tau$ in $\Pi|z$. It is unlikely that the execution $\tau$ catches us making the tiny local change $z\mapsto z^B$ in the input, and one expects that $\tau$ continues to appear in $\Pi|z^B$. (This is the usual \emph{corruption} property of large rectangles.) Formally, for windows $w_1$ and $w_2$, we say
\begin{equation}
\Pi|w_1\enspace \text{\itshape overflows onto}\enspace~\Pi|w_2
\qquad\text{iff}\qquad
\textstyle\sum_\tau\,\max(p^1_\tau-p^2_\tau,0)~\le~o(\alpha),
\end{equation}
where\footnote{Note that the event in $\Pr[\Pi|w_i=\tau]$ is to be parsed as ``a sample from the distribution $(\Pi|w_i)$ yields $\tau$''.} $p^i_\tau\coloneqq\Pr[\Pi|w_i=\tau]$ and the sum is over accepting transcripts $\tau$. (The definition of overflow makes sense for any distributions over transcripts; we will also apply it to $\Pi|z$.) For technical reasons (which will become apparent shortly), we shall flip two paths instead of one in order to pass from 3-violation inputs to 7-violation inputs.

\begin{three-vs-seven} \label{thm:3v7}
Let $z$ be any 3-violation input and let $\calP$ be any pairing of the non-violated terminals with canonical edge-disjoint paths $B_1,\ldots,B_{k-1}$. Let $w$ be a random window of~$z$, and choose distinct $i,j\in[k-1]$ at random. Then, with probability $\geq 1-o(1)$,
\[
\Pi|w \enspace\,\text{\itshape overflows onto}\enspace\,\Pi|w^{B_i\cup B_j}.
\]
\end{three-vs-seven}

\subsection{Deriving the contradiction}

We now prove \ref{implication} by applying the \onethree and the \threeseven in a black-box fashion to find some 7-violation input that $\Pi$ accepts with too high a probability $> 6\alpha$.

\begin{wrapfigure}[5]{r}{4.8cm}
\begin{lpic}[l(4mm),b(10mm),t(-5mm),r(-3mm)]{fano(.3)}
\lbl[c]{80 ,-10;$F=([7],E)$}

\lbl[c]{20 ,10;$1$}
\lbl[c]{80 ,10;$2$}
\lbl[c]{140 ,10;$3$}
\lbl[c]{118 ,77;$4$}
\lbl[c]{42 ,77;$6$}
\lbl[c]{88 ,125;$5$}
\lbl[c]{85 ,42;$7$}
\end{lpic}
\end{wrapfigure}

Define $F\coloneqq ([7],E)$ as the \emph{Fano plane} hypergraph on 7 nodes. See the figure on the right. This hypergraph has 7 hyperedges, each of which is incident to 3 nodes, and the hyperedges are pairwise uniquely intersecting. For each hyperedge $e\in E$ choose some arbitrary but fixed pairing $\calP^e$ of the remaining nodes in $[7]\smallsetminus e$.

\emph{Probability space.}
Choose the following at random:
\begin{enumerate}
\item An injection of $[7]$ into $T$. Denote the result by $v_1,\ldots,v_7\in T$.
\item A pairing $\calP$ of the remaining terminals $T\smallsetminus \{v_1,\ldots,v_7\}$.
\item A 7-violation input $z_7$ with $\viol(z_7)= \{v_1,\ldots,v_7\}$.
\item A window $w_7$ of $z_7$.
\end{enumerate}
We do not make a distinction between the nodes of $F$ and their embedding $\{v_1,\ldots,v_7\}$ in~$T$. In particular, we think of the hyperedges $e\in E$ as triples of terminals, and the $\calP^e$ as pairings of terminals. Associated with the pairing $\calP^e \cup \calP$ there is a canonical collection of edge-disjoint paths; let $\{B_1^e,B_2^e\}$ denote the two paths that connect $\calP^e$ in this collection.

Based on the above, we define seven $3$-violation windows, indexed by $e\in E$:
\[
\text{window}\enspace w_e\coloneqq w_7^{B_1^e\cup B_2^e}\enspace\text{of}\enspace
z_e\coloneqq z_7^{B_1^e\cup B_2^e}
\qquad\text{\itshape (note:$\enspace\,\viol(z_e)=e$)}.
\]
The following claim (proved at the end of this subsection) follows directly from the \onethree and the \threeseven as soon as we view our probability space from the right perspective.
\begin{claim} \label{clm:list}
In the following list of 28 events, each occurs with probability $\geq 1-o(1)$:
\begin{itemize}
\item {\slshape Overflow for $e\in E$:}~~$\Pi|w_e$ overflows onto $\Pi|w_7$.
\item {\slshape Near-disjointness for $\{e,e'\}\subseteq E$:}~~$\Delta\bigl(\Pia|w_e,\Pia|w_{e'}\bigr)\geq 1-o(1)$.
\end{itemize}
\end{claim}

By a union bound over all the 28 events in the above list, we can fix our random choices 1--4 to obtain a fixed 7-violation window $w_7$ and fixed 3-violation windows $w_e$ such that
\begin{alignat}{3}
\text{\slshape Overflow:}\quad
&\forall e\in E:
&\quad\textstyle \sum_\tau \max(p^e_\tau-p^7_\tau,0)~&\leq~o(\alpha),
\label{eq:overflow}\\
\text{\slshape Near-disjointness:}\quad
&\forall \{e,e'\}\subseteq E:
&\textstyle \sum_\tau \min(p^e_\tau,p^{e'}_\tau)~&\leq~o(\alpha).
\label{eq:almostdisjoint}
\end{alignat}
Here $p^7_\tau\coloneqq\Pr[\Pi|w_7=\tau]$, $p^e_\tau\coloneqq\Pr[\Pi|w_e=\tau]$, and the sums are over accepting transcripts; we have also rephrased the near-disjointness property using the fact that $\Pr[\Pi|w_e~\text{accepts}]=2\alpha$.

These two properties state that typical accepting transcripts for $\Pi|w_e$ contribute to the acceptance probability of $\Pi|w_7$, and these contributions are pairwise near-disjoint. Hence, roughly speaking, one expects $\Pr[\Pi|w_7~\text{accepts}]$ to be at least $\sum_{e\in E} \Pr[\Pi|w_e~\text{accepts}] = 7\cdot 2\alpha = 14\alpha >6\alpha$. But then some 7-violation input in $w_7$ would be accepted with probability $>6\alpha$, which completes the proof of \ref{implication} (and hence \autoref{thm:communication-lb}). Indeed, we perform this calculation carefully as follows. We first partition the set of accepting transcripts as $\bigcup_{e\in E} S_e$ where $S_e$ consists of those $\tau$'s for which $p^e_\tau=\max_{e'}p^{e'}_\tau$ (breaking ties arbitrarily). Then{\allowdisplaybreaks
\begin{align*}
\Pr[\Pi|w_7\text{ accepts}]~&\textstyle=~\sum_\tau p^7_\tau\\
&\textstyle\ge~\sum_{e\in E,\,\tau\in S_e}\,\min(p^7_\tau,p^e_\tau)\\
&\textstyle=~\sum_{e\in E,\,\tau\in S_e}\,\bigl(p^e_\tau-\max(p^e_\tau-p^7_\tau,0)\bigr)\\
&\textstyle\ge~\sum_{e\in E,\,\tau\in S_e}\,p^e_\tau-\sum_{e\in E,\,\tau}\,\max(p^e_\tau-p^7_\tau,0)\\
&\textstyle\ge~\sum_{e\in E,\,\tau\in S_e}\,p^e_\tau-7\cdot o(\alpha) \tag{via \eqref{eq:overflow}}\\
&\textstyle=~\sum_{e\in E,\,\tau}\,p^e_\tau-\sum_{e\in E,\,e'\in E\smallsetminus\{e\},\,\tau\in S_{e'}}\,p^e_\tau-o(\alpha)\\
&\textstyle=~\sum_{e\in E,\,\tau}\,p^e_\tau-\sum_{e\in E,\,e'\in E\smallsetminus\{e\},\,\tau\in S_{e'}}\,\min(p^e_\tau,p^{e'}_\tau)-o(\alpha)\\
&\textstyle\ge~\sum_{e\in E,\,\tau}\,p^e_\tau-\sum_{e\in E,\,e'\in E\smallsetminus\{e\},\,\tau}\,\min(p^e_\tau,p^{e'}_\tau)-o(\alpha)\\
&\textstyle\ge~\sum_{e\in E,\,\tau}\,p^e_\tau-7\cdot 6\cdot o(\alpha)-o(\alpha) \tag{via \eqref{eq:almostdisjoint}}\\
&\textstyle=~\sum_{e\in E}\,\Pr[\Pi|w_e\text{ accepts}]-o(\alpha) \\
&\textstyle=~7\cdot 2\alpha-o(\alpha)\\
&\textstyle=~(14-o(1))\cdot\alpha\\
&\textstyle>~6\alpha.
\end{align*}
}

\begin{proof}[Proof of \autoref{clm:list}]

\emph{Overflow.}
For notational convenience, suppose $e=\{v_1,v_2,v_3\}$ and $\calP^e=\{\{v_4,v_7\},\{v_5,v_6\}\}$. An alternative way to generate a sample from our probability space is (in steps 1 and 6, we are really picking random injections):
\begin{enumerate}
\item Random $\{v_1,v_2,v_3\}\subseteq T$.
\item Random 3-violation input $z_e$ subject to $\viol(z_e) = \{v_1,v_2,v_3\}$.
\item Random pairing $\calP'=\{P_1,\ldots,P_{k-1}\}$ of $T\smallsetminus \{v_1,v_2,v_3\}$ with canonical paths $B_1,\ldots,B_{k-1}$.
\item Random window $w_e$ of $z_e$.
\item Random distinct $i,j\in[k-1]$.
\item Random $\{v_4,v_7\} = P_i$ and $\{v_5,v_6\}= P_j$.
\item Deterministically, define $z_7\coloneqq z_e^{B_i\cup B_j}$ and $w_7 \coloneqq w_e^{B_i\cup B_j}$ and $\calP\coloneqq \calP'\smallsetminus \{P_i,P_j\}$.
\end{enumerate}
The choices made in steps 1--3 match the data that is quantified universally in the \threeseven, whereas steps 4 and 5 make random choices as in the \threeseven; hence the lemma applies.

\begin{figure}%
\centering
\begin{tikzpicture}[auto,scale=1,%
  myarrow/.style={%
		line width=.5mm,
    decorate,
    decoration={%
      snake,
      segment length=4mm,
      amplitude=0.6mm,
      pre length=0pt,
      post length=0pt,
    }
  }]
\tikzstyle{n} = [font=\large,inner sep=2,circle,draw,line width=0.2mm];
\tikzstyle{m} = [midway,black];

\begin{scope}[rotate=90,xscale=-.9]
\node[n] (v3)  at (0:3) {$v_3$};
\node[n] (v4)  at (-51:3) {$v_4$};
\node[n] (v5)  at (-103:3) {$v_5$};
\node[n] (v6)  at (-154:3) {$v_6$};
\node[n] (v7)  at (-206:3) {$v_7$};
\node[n] (v1)  at (-257:3) {$v_1$};
\node[n] (v2)  at (-309:3) {$v_2$};
\end{scope}

\draw[myarrow,YellowOrange] (v1) -- (v2) node [m,left] {$B'_i$};
\draw[myarrow,Fuchsia] (v4) -- (v5) node [m,right] {$B'_j$};

\draw[myarrow,JungleGreen] (v4) -- (v7) node [m,right,xshift=5,yshift=-2] {$B^e_1$};
\draw[myarrow,JungleGreen] (v5) -- (v6) node [m,right,yshift=6] {$B^e_2$};
\draw[myarrow,Bittersweet] (v2) -- (v6) node [m,left,xshift=-6] {$B^{e'}_2$};
\draw[myarrow,Bittersweet] (v1) -- (v7) node [m,left,yshift=7.5] {$B^{e'}_1$};
\end{tikzpicture}%
\hspace{2.5cm}%
\begin{tikzpicture}[auto,scale=.85]
\tikzstyle{n} = [font=\large,inner sep=2];
\tikzstyle{m} = [midway,above,sloped,black];
\tikzstyle{myarrow} = [<->,line width=.4mm];

\node at (-2.4,0) {};
\node[n] (seven)  at (0,3) {$w_7$};
\node[n] (e)      at (-2,0) {$w_e$};
\node[n] (f)      at (2,.5) {$w_{e'}$};
\node[n] (fhat)   at (2,-.5) {$\hat{w}_{e'}$};
\node[n] (one)      at (0,-3) {$w_1$};

\small

\draw[myarrow,JungleGreen] (seven) -- (e) node [m] {$B^e_1\cup B^e_2$};
\draw[myarrow,Bittersweet] (seven) -- (f) node [m] {$B^{e'}_1\cup B^{e'}_2$};
\draw[myarrow,YellowOrange] (e) -- (one) node [m,below] {$B'_i$};
\draw[myarrow,Fuchsia] (fhat) -- (one) node [m,below] {$B'_j$};
\draw[dashed,black] (1.4,0) -- (2.6,0);
\end{tikzpicture}

\caption{Illustration for the proof of \autoref{clm:list}. \emph{Left:} Paths flipped between terminals. \emph{Right:} Relationships between windows.}
\label{fig:flips}
\vspace{-1mm}
\end{figure}
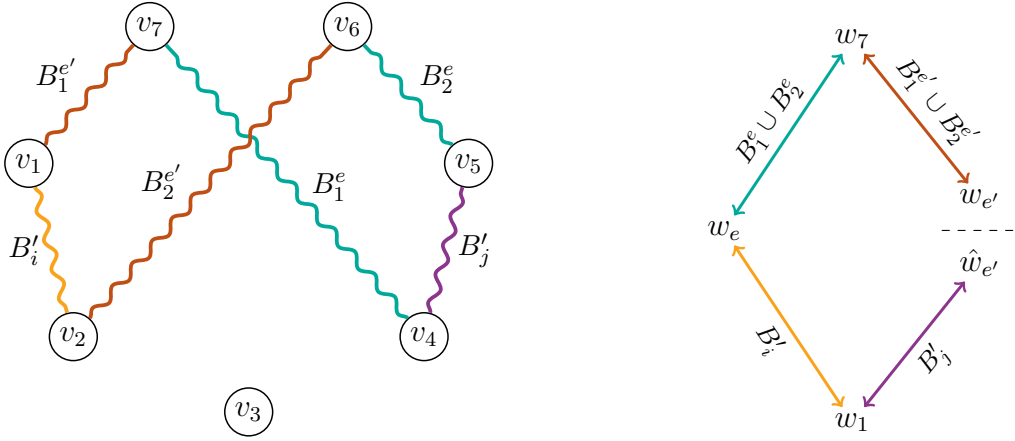

\emph{Near-disjointness.}
For notational convenience, suppose $e=\{v_1,v_2,v_3\}$, $e'=\{v_3,v_4,v_5\}$, $\calP^e=\{\{v_4,v_7\},\{v_5,v_6\}\}$, and $\calP^{e'}=\{\{v_1,v_7\},\{v_2,v_6\}\}$ (it does not matter for the proof how $\calP^e$ and $\calP^{e'}$ were chosen). An alternative way to generate a sample from our probability space is (see \autoref{fig:flips}):
\begin{enumerate}
\item Random $v_3\in T$.
\item Random 1-violation input $z_1$ subject to $\viol(z_1) = \{v_3\}$.
\item Random pairing $\calP'=\{P'_1,\ldots,P'_k\}$ of $T\smallsetminus \{v_3\}$ with canonical paths $B'_1,\ldots,B'_k$.
\item Random window $w_1$ of $z_1$.
\item Random distinct $i,j,l \in[k]$.
\item Random $\{v_1,v_2\} = P'_i$ and $\{v_4,v_5\}= P'_j$ and $\{v_6,v_7\}=P'_l$.\\[-4mm]
\item Deterministically, define
\begin{itemize}[label=$-$]
\item $z_e \coloneqq z_1^{B'_i}$ and $w_e \coloneqq w_1^{B'_i}$,
\item $\hat{z}_{e'} \coloneqq z_1^{B'_j}$ and $\hat{w}_{e'} \coloneqq w_1^{B'_j}$,
\item $\calP\coloneqq\calP'\smallsetminus\{P_i,P_j,P_l\}$,
\item $\{B^e_1,B^e_2\}$ according to the canonical paths for $\calP^e\cup \calP$,
\item $\{B^{e'}_1,B^{e'}_2\}$ according to the canonical paths for $\calP^{e'}\cup \calP$,
\item $z_7 \coloneqq z_e^{B^e_1\cup B^e_2}$ and $w_7 \coloneqq w_e^{B^e_1\cup B^e_2}$,
\item $z_{e'} \coloneqq z_7^{B^{e'}_1\cup B^{e'}_2}$ and $w_{e'} \coloneqq w_7^{B^{e'}_1\cup B^{e'}_2}$.
\end{itemize}
\end{enumerate}
The choices made in steps 1--3 match the data that is quantified universally in the \onethree, whereas steps 4 and 5 (excluding variable $l$) make random choices as in the \onethree. Hence that lemma applies and shows that $\Pia|w_e$ and $\Pia|\hat{w}_{e'}$ are near-disjoint with high probability. Finally, we note that $\hat{w}_{e'}$ and $w_{e'}$ differ by the flipping of an eulerian graph, namely $B'_j\oplus B'_i\oplus B^e_1\oplus B^e_2\oplus B^{e'}_1\oplus B^{e'}_2$ (where $\oplus$ means symmetric difference), so $\Pi|w_{e'}$ and $\Pi|\hat{w}_{e'}$ have the same distribution assuming w.l.o.g.\ that $\Pi$ is smooth (as discussed in \autoref{sec:preliminaries}). Thus $\Pia|w_e$ and $\Pia|w_{e'}$ are also near-disjoint with high probability.
\end{proof}

\subsection{Roadmap for the rest of the proof}

We prove the \onethree in \autoref{sec:technical:1vs3} and the \threeseven in \autoref{sec:technical:3vs7}. Both proofs rely on another technical lemma, the \homog (stated below, proved in \autoref{sec:technical:homogeneity}), which generalizes a lemma from (the full version of)~\cite[\S5]{huynh12virtue}. In fact, we prove the \homog for any gadget $g$ that is \emph{regular} (as defined in \autoref{sec:technical:homogeneity}), which our gadget is.

\begin{homogeneity} \label{thm:homog}
Fix an arbitrary $z\in\{0,1\}^m$ for some $m$. Let $W$ be a random window of $z$ in $g^m$, let $XY$ be a random input in $W$, and let $R$ be an arbitrary random variable that is conditionally independent of $W$ given $XY$. If $\I(R\semi XY\mid W)\le o(1)$ then at least a $1-o(1)$ fraction of windows $w$ of $z$ are such that $\Delta(R|w,R|z)\le o(1)$.
\end{homogeneity}

In the statement, $R|w$ is shorthand for $R|(W=w)$, and $R|z$ denotes the marginal distribution of $R$ in the whole probability space, which is over uniformly random $XY\in(g^m)^{-1}(z)$. Furthermore, we mention that our proof shows that at least a $1-o(1)$ fraction of $xy\in(g^m)^{-1}(z)$ are such that $\Delta(R|xy,R|z)\le o(1)$, but for the \onethree and the \threeseven we only require the property for windows.

In \autoref{sec:gadget} we defined the directed flips $w^\leftarrow,w^\nwarrow,w^\uparrow$ for a single-gadget window. We now also define directed flips for multi-gadget windows $w$: if $B$ is a subset of coordinates then $w^{\leftarrow B},w^{\nwarrow B},w^{\uparrow B}$ are defined by applying the corresponding directed flips to the coordinates in $B$. Then we have the following key property of our gadget.

\begin{fact} \label{fact:dirflip}
If $w$ is a uniformly random window of $z$, then each of $w^{\leftarrow B},w^{\nwarrow B},w^{\uparrow B}$ is marginally a uniformly random window of $z^B$.
\end{fact}

This concept is used in the proofs of the \onethree and the \threeseven. It turns out that the \threeseven can be proved (with a small modification to our proof) even for the simpler gadget that was used in \cite{sherstov11pattern,goos14communication} (as can the \homog since that gadget is regular), but our proof of the \onethree crucially uses \autoref{fact:dirflip}, which does not hold for that simpler gadget.

\subsection[Proof of the \onethreename]{Proof of the \onethree} \label{sec:technical:1vs3}

Consider a probability space with the following random variables: $I\in[k]$, $J\in[k]\smallsetminus\{I\}$, $W$ is a random window of $z^{B_I}$, $XY$ is a random input in $W$, and $\Pia$ is the random transcript of $\Pi$ on input $XY$ conditioned on acceptance. For convenience, denote $B\coloneqq B_1\cup\cdots\cup B_k$ and $B_{-i}\coloneqq B\smallsetminus B_i$. We have \[\I\bigl(\Pia\semi(XY)_{B_{-I}}\bigmid IW\bigr)~\le~\H(\Pia\mid IW)~\le~|\Pi|~\le~o(k)\] so by the standard direct sum property \cite{bar-yossef04information},
\begin{align*}
\I\bigl(\Pia\semi(XY)_{B_J}\bigmid IJW\bigr)~&\textstyle=~\frac{1}{k-1}\cdot\E_{i\sim I}\sum_{j\in[k]\smallsetminus\{i\}}\I\bigl(\Pia\semi(XY)_{B_j}\bigmid W,I=i\bigr)\\
&\textstyle\le~\frac{1}{k-1}\cdot\I\bigl(\Pia\semi(XY)_{B_{-I}}\bigmid IW\bigr)\\
&\le~o(1).
\end{align*}
Define $H\coloneqq\{I,J\}$, and abbreviate $B_I\cup B_J$ as $B_H$ and $W_{[n]\smallsetminus(B_I\cup B_J)}$ as $W_{-B_H}$. By Markov's inequality, with probability $\ge 1-o(1)$ over $h\sim H$ and $w_{-B_h}\sim W_{-B_h}$, we have \[\I\bigl(\Pia\semi(XY)_{B_J}\bigmid IJW_{B_h},H=h,W_{-B_h}=w_{-B_h}\bigr)~\le~o(1).\] Fixing such $h$ and $w_{-B_h}$ (henceforth), say $h=\{1,2\}$, it suffices to show that with probability $\ge 1-o(1)$ over a random window $w_{B_h}$ of $z_{B_h}$, we have $\Delta\bigl(\Pia|w^{B_1},\Pia|w^{B_2}\bigr)\ge 1-o(1)$ (where $w$ is the combination of $w_{B_h}$ and $w_{-B_h}$).

We rephrase the situation as follows. Consider a protocol $\Pi_*$ that interprets its input as $(xy)_{B_h}$, uses private coins to sample random $(xy)_{-B_h}$ from $w_{-B_h}$, and runs $\Pi$ on the input $xy$ (the combination of $(xy)_{B_h}$ and $(xy)_{-B_h}$). Henceforth recycling notation by letting $z\in\{0,1\}^{|B_h|}$ refer to $z_{B_h}$, and letting $(I,J)$ be random in $\{(1,2),(2,1)\}$, $W$ be a random window of (the new) $z^{B_I}$, and $XY$ be a random input to $\Pi_*$ in $W$, the situation is:
\begin{itemize}[leftmargin=3.25cm]
\item[\bf Assumption:~] $\I\bigl(\Pia_*\semi(XY)_{B_J}\bigmid IJW\bigr)\le o(1)$.
\item[\bf Want to show:~] For $\ge 1-o(1)$ fraction of windows $w$ of $z$, $\Delta\bigl(\Pia_*|w^{B_1},\Pia_*|w^{B_2}\bigr)\ge 1-o(1)$.
\end{itemize}

The assumption holds (with factor $2$ loss in the $o(1)$) conditioned on either outcome of $(I,J)$; let us tacitly condition on the outcome $(1,2)$. Then $\I\bigl(\Pia_*\semi(XY)_{B_2}\bigmid W\bigr)\le o(1)$ where $W$ is a random window of $z^{B_1}$. By Markov's inequality, with probability $\ge 1-o(1)$ over $w_{B_1}\sim W_{B_1}$ we have $\I\bigl(\Pia_*\semi(XY)_{B_2}\bigmid W_{B_2},W_{B_1}=w_{B_1}\bigr)\le o(1)$; call such a $w_{B_1}$ \emph{good}. Hence for a good $w_{B_1}$, we can apply the \homog with $m\coloneqq|B_2|$ and $R\coloneqq\Pia_*|(W_{B_1}=w_{B_1})$ (note that $R|(xy)_{B_2}$ is the distribution of $\Pia_*$ on input $(XY)_{B_1}(xy)_{B_2}$ where $(XY)_{B_1}$ is random in $w_{B_1}$). This tells us that for a good $w_{B_1}$, with probability $\ge 1-o(1)$ over $w_{B_2}\sim W_{B_2}$ we have $\Delta\bigl(\Pia_*|w_{B_1}w_{B_2},\Pia_*|w_{B_1}z_{B_2}\bigr)\le o(1)$, where the distribution $\Pia_*|w_{B_1}z_{B_2}$ is over random $(XY)_{B_1}\in w_{B_1}$ and $(XY)_{B_2}\in (g^m)^{-1}(z_{B_2})$. We summarize the above with the following claim.

\begin{claim} \label{clm:window1}
For $\ge 1-o(1)$ fraction of windows $w$ of $z^{B_1}$, we have $\Delta\bigl(\Pia_*|w,\Pia_*|w_{B_1}z_{B_2}\bigr)\le o(1)$.
\end{claim}

Conditioning on the other outcome $(I,J)=(2,1)$ yields the symmetric property.

\begin{claim} \label{clm:window2}
For $\ge 1-o(1)$ fraction of windows $w$ of $z^{B_2}$, we have $\Delta\bigl(\Pia_*|w,\Pia_*|z_{B_1}w_{B_2}\bigr)\le o(1)$.
\end{claim}

Now pick a random window $w$ of $z^{B_h}$. Using \autoref{fact:dirflip}, $w^{B_2}$ and $w^{\nwarrow B_2}$ are both uniformly random (albeit correlated) windows of $z^{B_1}$, and $w^{B_1}$ and $w^{\nwarrow B_1}$ are both uniformly random (albeit correlated) windows of $z^{B_2}$. Hence by \autoref{clm:window1}, \autoref{clm:window2}, and a union bound, with probability $\ge 1-o(1)$ over the choice of $w$, the following four distances are simultaneously $\le o(1)$: $\Delta\bigl(\Pia_*|w^{B_2},\Pia_*|w_{B_1}z_{B_2}\bigr)$, $\Delta\bigl(\Pia_*|w^{\nwarrow B_2},\Pia_*|w_{B_1}z_{B_2}\bigr)$, $\Delta\bigl(\Pia_*|w^{B_1},\Pia_*|z_{B_1}w_{B_2}\bigr)$, $\Delta\bigl(\Pia_*|w^{\nwarrow B_1},\Pia_*|z_{B_1}w_{B_2}\bigr)$.

We argue shortly that $\Delta\bigl(\Pia_*|w^{\nwarrow B_1},\Pia_*|w^{\nwarrow B_2}\bigr)=1$ with probability $1$; putting everything together then shows that $\Delta\bigl(\Pia_*|w^{B_1},\Pia_*|w^{B_2}\bigr)\ge 1-o(1)$, as illustrated below. (This is equivalent to what we want to show, since sampling a window $w$ of $z^{B_h}$ and taking $w^{B_1},w^{B_2}$ is equivalent to sampling a window $w$ of $z$ and taking $w^{B_2},w^{B_1}$.)

\begin{figure}[H]
\centering
\begin{tikzpicture}
\tikzset{outer sep=4}
\node (a) at (0,2.5) {$\Pia_*|w^{B_1}$};
\node (b) at (5,2.5) {$\Pia_*|z_{B_1}w_{B_2}$};
\node (c) at (10,2.5) {$\Pia_*|w^{\nwarrow B_1}$};
\node (d) at (0,0) {$\Pia_*|w^{B_2}$};
\node (e) at (5,0) {$\Pia_*|w_{B_1}z_{B_2}$};
\node (f) at (10,0) {$\Pia_*|w^{\nwarrow B_2}$};
\path[<->,semithick]
(a) edge node[above] {\footnotesize $\Delta\le o(1)$} (b)
(b) edge node[above] {\footnotesize $\Delta\le o(1)$} (c)
(d) edge node[below] {\footnotesize $\Delta\le o(1)$} (e)
(e) edge node[below] {\footnotesize $\Delta\le o(1)$} (f)
(c) edge node[right] {\footnotesize $\Delta=1$} (f)
(a) edge[dashed] node[left] {\footnotesize $\Delta\ge 1-o(1)$} (d);
\end{tikzpicture}
\end{figure}

To finish the proof, suppose for contradiction that some accepting transcript has positive probability under both $\Pia_*|xy$ and $\Pia_*|x'y'$ for some $xy\in w^{\nwarrow B_1}$ and $x'y'\in w^{\nwarrow B_2}$. Then $\Pi_*$ would also accept $xy'$ with positive probability. We claim that $g^{|B_h|}(xy')=z$. To see this, consider any coordinate $c$ of $z$; suppose $c\in B_1$ (the case $c\in B_2$ is similar). There is an embedding of stretched $\AND$ (if $z_c=0$) or $\NAND$ (if $z_c=1$) such that $w^{\nwarrow B_1}_c$ is the image of $(0,0)$ (hence is $z_c$-monochromatic) and $w^{\nwarrow B_2}_c=w_c$ is the image of $(1,1)$ (hence is $(1-z_c)$-monochromatic). Since $(xy)_c\in w^{\nwarrow B_1}_c$ and $(x'y')_c\in w_c$, it follows that $(xy')_c$ is in the image of $(0,1)$, which is $z_c$-monochromatic. So $g((xy')_c)=z_c$ and the claim is proved.

Since $\Pi_*$ accepts some input in $(g^{|B_h|})^{-1}(z)$ with positive probability (for the new $z$), it follows that $\Pi$ accepts some input in $(g^n)^{-1}(z)$ with positive probability, for the original $z$, which is a contradiction since the original $z$ has only one violation.

\subsection[Proof of the \threesevenname]{Proof of the \threeseven} \label{sec:technical:3vs7}

Assume for convenience that $k-1$ is even. Note that sampling distinct $i,j\in[k-1]$ is equivalent to sampling a permutation $\sigma$ of $[k-1]$ and an $h\in[\frac{k-1}{2}]$ and setting $i=\sigma(2h-1)$, $j=\sigma(2h)$.

Thus we have a probability space with random variables $\Sigma,H,I,J$ corresponding to the above, as well as the following: $W$ is a random window of $z$, $XY$ is a random input in $W$, and $\Pia$ is the random transcript of $\Pi$ on input $XY$ conditioned on acceptance. For convenience, denote $B\coloneqq B_1\cup\cdots\cup B_{k-1}$ and $B_{ij}\coloneqq B_i\cup B_j$. We have \[\I\bigl(\Pia\semi(XY)_B\bigmid W\bigr)~\le~\H(\Pia\mid W)~\le~|\Pi|~\le~o(k)\] so by the standard direct sum property \cite{bar-yossef04information},{\allowdisplaybreaks
\begin{align*}
\I\bigl(\Pia\semi(XY)_{B_{IJ}}\bigmid WIJ\bigr)~&=~\I\bigl(\Pia\semi(XY)_{B_{IJ}}\bigmid W\Sigma H\bigr)\\
&\textstyle=~\frac{2}{k-1}\cdot\sum_{h\in[(k-1)/2]}\I\bigl(\Pia\semi(XY)_{B_{IJ}}\bigmid W\Sigma,H=h\bigr)\\
&\textstyle\le~\frac{2}{k-1}\cdot\I\bigl(\Pia\semi(XY)_B\bigmid W\Sigma\bigr)\\
&\textstyle=~\frac{2}{k-1}\cdot\I\bigl(\Pia\semi(XY)_B\bigmid W\bigr)\\
&\le~o(1).
\end{align*}
}Abbreviate $W_{[n]\smallsetminus B_{ij}}$ as $W_{-B_{ij}}$. By Markov's inequality, with probability $\ge 1-o(1)$ over $ij\sim IJ$ and $w_{-B_{ij}}\sim W_{-B_{ij}}$, we have $\I\bigl(\Pia\semi(XY)_{B_{ij}}\bigmid W_{B_{ij}},W_{-B_{ij}}=w_{-B_{ij}}\bigr)\le o(1)$. Fixing such $ij$ and $w_{-B_{ij}}$ (henceforth), it suffices to show that with probability $\ge 1-o(1)$ over $w_{B_{ij}}\sim W_{B_{ij}}$, $\Pi|w$ overflows onto $\Pi|w^{B_{ij}}$ (where $w$ is the combination of $w_{B_{ij}}$ and $w_{-B_{ij}}$).

We rephrase the situation as follows. Consider a protocol $\Pi_*$ that interprets its input as $(xy)_{B_{ij}}$, uses private coins to sample random $(xy)_{-B_{ij}}$ from $w_{-B_{ij}}$, and runs $\Pi$ on the input $xy$ (the combination of $(xy)_{B_{ij}}$ and $(xy)_{-B_{ij}}$). Henceforth recycling notation by letting $z\in\{0,1\}^{|B_{ij}|}$ refer to $z_{B_{ij}}$, letting $B$ refer to $B_{ij}$, and letting $W$ be a random window of (the new) $z$ and $XY$ be a random input to $\Pi_*$ in $W$, the situation is:
\begin{itemize}[leftmargin=3.25cm]
\item[\bf Assumption:~] $\I\bigl(\Pia_*\semi XY\bigmid W\bigr)\le o(1)$.
\item[\bf Want to show:~] For $\ge 1-o(1)$ fraction of windows $w$ of $z$, $\Pi_*|w$ overflows onto $\Pi_*|w^B$.
\end{itemize}

\begin{claim} \label{clm:overflow}
For $\ge 1-o(1)$ fraction of windows $w$ of $z^B$, $\Pi_*|z$ overflows onto $\Pi_*|w$.
\end{claim}

We prove \autoref{clm:overflow} shortly, but first we finish the proof of the \threeseven assuming it. By the \homog (with $m\coloneqq |B|$ and $R\coloneqq\Pia_*$), \autoref{clm:overflow}, and a union bound, at least a $1-o(1)$ fraction of windows $w$ of $z$ are such that both $\Delta\bigl(\Pia_*|w,\Pia_*|z\bigr)\le o(1)$ and $\Pi_*|z$ overflows onto $\Pi_*|w^B$ (since $w^B$ is a uniform window of $z^B$ if $w$ is a uniform window of $z$). We show that this implies that $\Pi_*|w$ overflows onto $\Pi_*|w^B$ as follows (letting $p^z_\tau$, $p^w_\tau$, $p^{w^B}_\tau$ denote the probability of a transcript $\tau$ under the distributions $\Pi_*|z$, $\Pi_*|w$, $\Pi_*|w^B$ respectively, and summing only over accepting $\tau$'s): \[\textstyle\sum_\tau\max(p^w_\tau-p^{w^B}_\tau,0)~\le~\sum_\tau\max(p^z_\tau-p^{w^B}_\tau,0)+\sum_\tau|p^w_\tau-p^z_\tau|~\le~o(\alpha)+o(\alpha)~=~o(\alpha).\]

\begin{proof}[Proof of \autoref{clm:overflow}]
By \autoref{fact:dirflip}, if $w$ is a random window of $z^B$, then $w^{\leftarrow B}$, $w^{\nwarrow B}$, $w^{\uparrow B}$ are each marginally uniformly random windows of $z$. Thus by the \homog (with $m\coloneqq |B|$ and $R\coloneqq\Pia_*$) and a union bound, with probability $\ge 1-o(1)$ over the choice of $w$, the following three distances are simultaneously $\le o(1)$: $\Delta\bigl(\Pia_*|w^{\leftarrow B},\Pia_*|z\bigr)$, $\Delta\bigl(\Pia_*|w^{\nwarrow B},\Pia_*|z\bigr)$, $\Delta\bigl(\Pia_*|w^{\uparrow B},\Pia_*|z\bigr)$. Now assuming this good event occurs for some particular $w$, we just need to show that $\Pi_*|z$ overflows onto $\Pi_*|w$.

(See \autoref{fig:overflow} for a proof-by-picture.) Let $p_\tau$, $p^{11}_\tau$, $p^{10}_\tau$, $p^{00}_\tau$, $p^{01}_\tau$ denote the probabilities of a transcript $\tau$ under $\Pi_*|z$, $\Pi_*|w$, $\Pi_*|w^{\leftarrow B}$, $\Pi_*|w^{\nwarrow B}$, $\Pi_*|w^{\uparrow B}$ respectively. Let $\gamma^{00}_\tau\coloneqq|p_\tau-p^{00}_\tau|$, and for $ab\in\{01,10\}$ let $\gamma^{ab}_\tau\coloneqq|p^{00}_\tau-p^{ab}_\tau|$. We claim that for all $\tau$, $p_\tau-p^{11}_\tau\le\gamma^{00}_\tau+\gamma^{01}_\tau+\gamma^{10}_\tau$; this will finish the proof since then (summing only over accepting $\tau$'s) \[\textstyle\sum_\tau\max(p_\tau-p^{11}_\tau,0)~\le~\sum_\tau(\gamma^{00}_\tau+\gamma^{01}_\tau+\gamma^{10}_\tau)~\le~o(\alpha)+o(\alpha)+o(\alpha)~=~o(\alpha)\] where the second inequality is because $\sum_\tau\gamma^{00}_\tau$, $\sum_\tau\gamma^{01}_\tau$, $\sum_\tau\gamma^{10}_\tau\le o(\alpha)$ follow from (respectively) $\Delta\bigl(\Pia_*|z,\Pia_*|w^{\nwarrow B}\bigr)$, $\Delta\bigl(\Pia_*|w^{\nwarrow B},\Pia_*|w^{\uparrow B}\bigr)$, $\Delta\bigl(\Pia_*|w^{\nwarrow B},\Pia_*|w^{\leftarrow B}\bigr)\le o(1)$.

\begin{figure}[t]
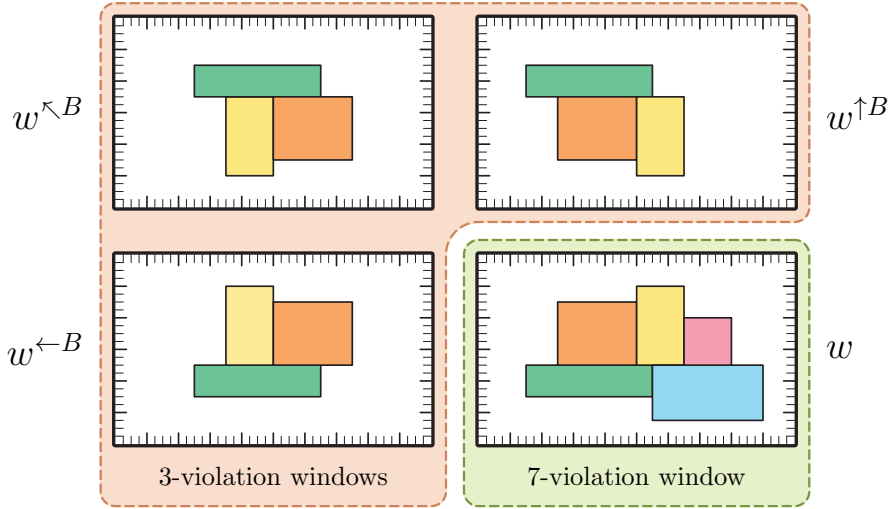
%
\centering
\begin{lpic}[t(-5mm)]{overflow(.21)}
\Large
\lbl[r]{0,110;$w^{\leftarrow B}$}
\lbl[r]{0,260;$w^{\nwarrow B}$}
\lbl[l]{470,110;$w$}
\lbl[l]{470,260;$w^{\uparrow B}$}
\small
\lbl[c]{120,30;3-violation windows}
\lbl[c]{350,30;7-violation window}
\end{lpic}
\caption{
Proof of \autoref{clm:overflow} illustrated. The four windows $w$, $w^{\leftarrow B}$, $w^{\protect\nwarrow B}$, $w^{\uparrow B}$ are rectangles of~$(x,y)$'s. Each $(x,y)$ can be further subdivided according to the private coins $(r_\A,r_\B)$ of the players. The protocol $\Pi_*$ partitions the extended input space of $(xr_\A,yr_\B)$'s into \emph{transcript rectangles}---above, we have only drawn \emph{accepting} transcript rectangles (in various colors). For a window $w'$, the probability $\Pr[\Pi_*|w' = \tau]$ is simply the \emph{area} (appropriately scaled) of the transcript rectangle of $\tau$ inside $w'$. In the proof of \autoref{clm:overflow}, the relevant case is when all of $\Pia_*|w^{\leftarrow B}$, $\Pia_*|w^{\protect\nwarrow B}$, $\Pia_*|w^{\uparrow B}$ have roughly the same distribution, say, $D$ (in fact, $D\coloneqq \Pia_*|z$). By the rectangular property of transcripts, this forces $\Pi_*|z$ to \emph{overflow onto} $\Pi_*|w$. (Note that $\Pia_*|w$ may contain additional transcripts to those in $D$, since the acceptance probability is higher.)}
\label{fig:overflow}
\end{figure}

To verify the subclaim, it suffices to show that
\begin{equation} \label{eq:diff}
p^{01}_\tau\cdot p^{10}_\tau~\ge~(p^{00}_\tau)^2-p^{00}_\tau\gamma^{01}_\tau-p^{00}_\tau\gamma^{10}_\tau
\end{equation}
since by the rectangular nature of transcripts, we have $p^{00}_\tau\cdot p^{11}_\tau=p^{01}_\tau\cdot p^{10}_\tau$, and thus if $p^{00}_\tau>0$ then \[p^{11}_\tau~=~\frac{p^{01}_\tau\cdot p^{10}_\tau}{p^{00}_\tau}~\ge~p^{00}_\tau-\gamma^{01}_\tau-\gamma^{10}_\tau~\ge~p_\tau-\gamma^{00}_\tau-\gamma^{01}_\tau-\gamma^{10}_\tau\] and if $p^{00}_\tau=0$ then of course $p^{11}_\tau\ge p^{00}_\tau=p_\tau-\gamma^{00}_\tau$. To see \eqref{eq:diff}, note that for some signs $\sigma^{01}_\tau,\sigma^{10}_\tau\in\{1,-1\}$, the left side of \eqref{eq:diff} equals $\bigl(p^{00}_\tau+\sigma^{01}_\tau\gamma^{01}_\tau\bigr)\cdot\bigl(p^{00}_\tau+\sigma^{10}_\tau\gamma^{10}_\tau\bigr)$, which expands to
\begin{equation} \label{eq:alt}
(p^{00}_\tau)^2+\sigma^{01}_\tau p^{00}_\tau\gamma^{01}_\tau+\sigma^{10}_\tau p^{00}_\tau\gamma^{10}_\tau+\sigma^{01}_\tau\sigma^{10}_\tau\gamma^{01}_\tau\gamma^{10}_\tau.
\end{equation}
If $\sigma^{01}_\tau=\sigma^{10}_\tau$ then \eqref{eq:alt} is at least the right side of \eqref{eq:diff} since the last term of \eqref{eq:alt} is nonnegative. If $\sigma^{01}_\tau\ne\sigma^{10}_\tau$, say $\sigma^{01}_\tau=-1$ and $\sigma^{10}_\tau=1$, then \eqref{eq:alt} is at least the right side of \eqref{eq:diff} since the sum of the last two terms in \eqref{eq:alt} is $p^{00}_\tau\gamma^{10}_\tau-\gamma^{01}_\tau\gamma^{10}_\tau=p^{01}_\tau\gamma^{10}_\tau\ge 0$.
\end{proof}

\subsection[Proof of the \homogname]{Proof of the \homog} \label{sec:technical:homogeneity}

\begin{definition} \label{def:gadget-graph}
For a gadget $g\colon\calX\times\calY\to\{0,1\}$ and $b\in\{0,1\}$, define the digraph $\calG^b$ as follows: the nodes are the $b$-inputs of $g$, and there is an edge from $xy$ to $x'y'$ iff $x=x'$ or $y=y'$. (That is, each node has a self-loop, and all $b$-inputs in a given row or column have all possible edges between them.)
\end{definition}

\begin{definition} \label{def:regular}
We say a gadget $g\colon\calX\times\calY\to\{0,1\}$ is \emph{regular} iff (i) $|\calX|=|\calY|$ is even, (ii) each row and each column is balanced (half $0$'s and half $1$'s), and (iii) $\calG^0$ and $\calG^1$ are both strongly connected.
\end{definition}

Our gadget $g$ is indeed regular, but we proceed to prove the lemma for any regular $g$.

The first part of the proof is inspired by a similar approach that was used in~\cite{huynh12virtue}. We augment the probability space with the following random variables: let $X'Y'$ be a random input in $W$ that is conditionally independent of $XY$ given $W$, and let $E\in((g^m)^{-1}(z))^2$ be chosen randomly from $\{(XY,X'Y'),(X'Y',XY)\}$. We have $\H(R\mid E)=\H(R\mid WE)\le\H(R\mid W)$ since $R$ is conditionally independent of $W$ given $E$, and conditioning decreases entropy. We also have $\H(R\mid XYE)=\H(R\mid XY)=\H(R\mid XYW)$ since $R$ is conditionally independent of $WE$ given $XY$. Putting these together, we get \[\I(R\semi XY\mid E)~=~\H(R\mid E)-\H(R\mid XYE)~\le~\H(R\mid W)-\H(R\mid XYW)~=~\I(R\semi XY\mid W)~\le~o(1).\] By Markov's inequality, with probability $\ge 1-o(1)$ over $e\sim E$, we have $\I(R\semi XY\mid E=e)\le o(1)$, in which case if $e=(x^{(0)}y^{(0)},x^{(1)}y^{(1)})$ then by Pinsker's inequality\footnote{Specifically, if $RB$ are jointly distributed random variables where $B\in\{0,1\}$ is a uniformly random bit, and $R_b$ denotes the distribution of $R|(B=b)$, then $\I(R\semi B)=\D(R_0\midd R)/2+\D(R_1\midd R)/2\ge 2\cdot(\Delta(R_0,R)^2/2+\Delta(R_1,R)^2/2)\ge 2\cdot(\Delta(R_0,R)/2+\Delta(R_1,R)/2)^2\ge\Delta(R_0,R_1)^2/2$, where $\D$ denotes KL-divergence, and the first inequality is Pinsker's, the second is by convexity of the square function, and the third is by the triangle inequality.}, $\Delta\bigl(R|x^{(0)}y^{(0)},R|x^{(1)}y^{(1)}\bigr)\le o(1)$; let us use $\epsilon>0$ for the latter $o(1)$ quantity. We describe what the above means in graph theoretic terms.

\begin{wrapfigure}[14]{r}{4.5cm}
\begin{lpic}[l(7mm),b(10mm),t(7mm),r(0mm)]{digraph(.35)}
\small
\lbl[c]{55 ,-9;Example of $\calG^1$ for}
\lbl[c]{55 ,-21.5;the regular gadget}
\lbl[c]{55 ,-33;$x_1+y_1+x_2y_2$}
\end{lpic}
\end{wrapfigure}

Define the digraph $\calG^z$ as follows: the nodes are the inputs in $(g^m)^{-1}(z)$, and there is an edge from one input to another iff there exists a window of $z$ containing both inputs; this includes a self-loop at each node. Note that $\calG^z$ is the tensor product $\calG^{z_1}\otimes\cdots\otimes\calG^{z_m}$, i.e., each node of $\calG^z$ corresponds to an $m$-tuple of nodes from those digraphs, and each edge of $\calG^z$ corresponds to an $m$-tuple of edges. For convenience, we make the dependence of the random variable $E$ on $z$ explicit using the notation $E^z$; thus $E^z$ is distributed over the edges of~$\calG^z$. By regularity, for $b\in\{0,1\}$ the distribution of $E^b$ over the edges of $\calG^b$ puts half its mass uniformly over the self-loops, and half its mass uniformly over the non-self-loops. Note that the distribution of $E^z$ is the product of the distributions of $E^{z_1},\ldots,E^{z_m}$, i.e., $E^z$ can be sampled by taking samples $(x^{(0,i)}y^{(0,i)},x^{(1,i)}y^{(1,i)})$ from $E^{z_i}$ (independent over $i\in[m]$) and forming the edge $\bigl(x^{(0,1)}y^{(0,1)}\cdots x^{(0,m)}y^{(0,m)},x^{(1,1)}y^{(1,1)}\cdots x^{(1,m)}y^{(1,m)}\bigr)$ in $\calG^z$.

We say an edge $(x^{(0)}y^{(0)},x^{(1)}y^{(1)})$ of $\calG^z$ is \emph{great} iff $\Delta\bigl(R|x^{(0)}y^{(0)},R|x^{(1)}y^{(1)}\bigr)\le\epsilon$. Thus the great edges have at least $1-o(1)$ probability mass under $E^z$.

Let $L$ be the number of non-self-loop edges in $\calG^b$ (which is the same for $b=0$ and $b=1$).

\begin{claim} \label{clm:walk}
There exists a distribution over length-$2L$ walks on $\calG^z$ such that (i) the first and last nodes are independent and each marginally uniform, and (ii) each of the $2L$ edges on the walk is marginally distributed according to $E^z$.
\end{claim}

\begin{proof}
By the product structure of $\calG^z$ and $E^z$, it suffices to prove this claim for a bit $b$ instead of $z$ (as the claim for $z$ follows by sampling $m$ independent such walks on the $\calG^{z_i}$'s and running them ``in parallel''). By regularity, if we ignore the self-loops, there exists an eulerian tour in $\calG^b$ that uses all the non-self-loop edges exactly once, and pays an equal number of visits to each node. Let $v_0,v_1,\ldots,v_{L-1},v_0$ denote the sequence of nodes visited (with repeats) on a fixed such tour. We explicitly describe the distribution of walks $v_{i_0},\ldots,v_{i_{2L}}$ on $\calG^b$, using mod-$L$ arithmetic:
\begin{itemize}
\item[1.] Independently sample $i_0$ and $\ell$ uniformly from $\{0,\ldots,L-1\}$.
\item[2.] For $j=1,\ldots,\ell$, execute one of the following with probability $1/2$ each:
\begin{itemize}[topsep=0pt]
\item[2a.] Use the self-loop then move forward (i.e., $i_{2j-1}=i_{2j-2}$ and $i_{2j}=i_{2j-1}+1$).
\item[2b.] Move forward then use the self-loop (i.e., $i_{2j-1}=i_{2j-2}+1$ and $i_{2j}=i_{2j-1}$).
\end{itemize}
\item[3.] For $j=\ell+1,\ldots,L$, execute one of the following with probability $1/2$ each:
\begin{itemize}[topsep=0pt]
\item[3a.] Use the self-loop twice (i.e., $i_{2j}=i_{2j-1}=i_{2j-2}$).
\item[3b.] Move forward then backward (i.e., $i_{2j-1}=i_{2j-2}+1$ and $i_{2j}=i_{2j-1}-1$).
\end{itemize}
\end{itemize}
This procedure has $L$ phases, each taking $2$ steps of the walk. Each of the first $\ell$ phases has the effect of moving forward one node on the tour, and each of the last $L-\ell$ phases has the effect of ending up at the same node the phase started at. Thus $i_{2L}=i_0+\ell$ and is hence independent of $i_0$ and uniform over $\{0,\ldots,L-1\}$ (since $\ell$ is independent of $i_0$ and uniform); hence also $v_{i_0}$ and $v_{i_{2L}}$ are independent and uniform (since the tour visits each node equally often) and so (i) is verified. Property (ii) holds even conditioned on any $\ell$, and can be verified by a little case analysis; e.g., if $\ell>1$ then the first edge is $(v_{i_0},v_{i_0})$ with probability $1/2$, and is $(v_{i_0},v_{i_0+1})$ with probability $1/2$ (this is a sample from $E^b$ since $v_{i_0}$ is a uniform node and $(v_{i_0},v_{i_0+1})$ is a uniform non-self-loop edge).
\end{proof}

If we sample a walk $x^{(0)}y^{(0)},\ldots,x^{(2L)}y^{(2L)}$ in $\calG^z$ as in \autoref{clm:walk}, then by property (ii) and a union bound, with probability $\ge 1-2L\cdot o(1)=1-o(1)$, each of the edges on the walk is great, in which case by the triangle inequality, $\Delta\bigl(R|x^{(0)}y^{(0)},R|x^{(2L)}y^{(2L)}\bigr)\le 2L\epsilon$. In summary, by property (i), a $1-o(1)$ fraction of pairs of inputs in $(g^m)^{-1}(z)$ are \emph{good} in the sense that their conditional distributions of $R$ are within statistical distance $2L\epsilon=o(1)$. Thus a $1-o(1)$ fraction of inputs $xy\in(g^m)^{-1}(z)$ are such that $(xy,\overline{xy})$ is good for a $1-o(1)$ fraction of $\overline{xy}\in(g^m)^{-1}(z)$, in which case (letting $\overline{xy}$ be random in $(g^m)^{-1}(z)$ in the following){\allowdisplaybreaks
\begin{align*}
\Delta(R|xy,R)~&\textstyle=~\Delta\bigl(R|xy,\E_{\overline{xy}}R|\overline{xy}\bigr)\\
&\textstyle\le~\E_{\overline{xy}}\,\Delta\bigl(R|xy,R|\overline{xy}\bigr)\\
&\textstyle\le~\Pr_{\overline{xy}}[(xy,\overline{xy})\text{ is good}]\cdot o(1)+\Pr_{\overline{xy}}[(xy,\overline{xy})\text{ is not good}]\cdot 1\\
&\le~1\cdot o(1)+o(1)\cdot 1\\
&=~o(1)
\end{align*}
}where the second line is a basic general fact about statistical distance. Say $xy$ is \emph{typical} if $\Delta(R|xy,R)\le o(1)$ as above. Note that in the original probability space, $XY$ is marginally uniform over $(g^m)^{-1}(z)$ and thus with probability at least $1-o(1)$ over sampling $w\sim W$ and $xy\sim XY\in w$, $xy$ is typical. It follows that for at least $1-o(1)$ fraction of $w$, at least $1-o(1)$ fraction of $xy\in w$ are typical, in which case{
\begin{align*}
\Delta(R|w,R)~&\textstyle=~\Delta\bigl(\E_{xy\in w}R|xy,R\bigr)\\
&\textstyle\le~\E_{xy\in w}\,\Delta(R|xy,R)\\
&\textstyle\le~\Pr_{xy\in w}[xy\text{ is typical}]\cdot o(1)+\Pr_{xy\in w}[xy\text{ is not typical}]\cdot 1\\
&\le~1\cdot o(1)+o(1)\cdot 1\\
&=~o(1).
\end{align*}

}\section{Query Lower Bound} \label{sec:query-lb}

An alternative approach for proving a lower bound for the \witness-game for $\TSE_G\circ g^n$ is:
\begin{itemize}[leftmargin=1.75cm]
\item[\emph{Step~1}:] Prove an appropriate \emph{query complexity} lower bound for $\TSE_G$.
\item[\emph{Step~2}:] Use a query-to-communication simulation theorem like~\cite{chan13approximate,goos15rectangles,lee15lower}.
\end{itemize}
In this section, we carry out the first step by proving an optimal $\Omega(n)$ lower bound (which in particular answers a question from~\cite{lovasz95search})---this proof is a lot simpler than our proof for the~$\Omega(n/\log n)$ communication lower bound in \autoref{sec:communication-lb}. Unfortunately, as we discuss below, it is not known how to perform the second step for constant-size gadgets $g$.

The result of this section can be interpreted as evidence that the right bound in \autoref{thm:main} is~$2^{\Omega(n)}$ and the right bound in \autoref{cor:expander} is $\Omega(n)$, and also as motivation for further work to improve parameters for simulation theorems.

\subsection{Query-to-communication}
The query complexity analogue of nonnegative rank decompositions (nonnegative combinations of nonnegative rank-1 matrices) are \emph{conical juntas}: nonnegative combinations of conjunctions of literals (input bits or their negations). We write a conical junta as $h=\sum_C w_C C$ where $w_C\geq 0$ and~$C$ ranges over all conjunctions $C\colon\{0,1\}^n\to\{0,1\}$. The \emph{degree} of $h$ is the maximum number of literals in a conjunction~$C$ with $w_C>0$. Each conical junta naturally computes a nonnegative function $h\colon \{0,1\}^n\to\R_{\geq 0}$. Hence we may study \witness-games in query complexity. In particular, the query complexity of the \witness-game for $\TSE_G$ is the least degree of a conical junta $h$ that on input~$z$ outputs~$h(z)=|\!\viol(z)|-1$.

The main result of~\cite{goos15rectangles} is a simulation of randomized protocols (or nonnegative rank decompositions) by conical juntas: a cost-$d$ protocol for a lifted problem $F\circ g^n$ can be simulated by a degree-$O(d)$ conical junta (approximately) computing $F$. While $F$ here is arbitrary, the result unfortunately assumes that $g\coloneqq\IP_b$ is a logarithmic-size, $b\coloneqq\Theta(\log n)$, inner-product function $\IP_b\colon\{0,1\}^b\times\{0,1\}^b\to\{0,1\}$ given by $\IP_b(x,y)\coloneqq\langle x,y\rangle\bmod{2}$.

Plugging $b$-bit gadgets into the reductions of \autoref{sec:reductions} would blow up the number of input bits of CSP-SAT exponentially in $b$. This is not only an artifact of our particular reduction! Consider more generally any reduction from a communication search problem $S\circ g^n$ to a $\KWp$-game for a monotone $f\colon\{0,1\}^m\to\{0,1\}$. Since the $\KWp$-game has \emph{nondeterministic} communication complexity $\log m$ (number of bits the players must nondeterministically guess to find a witness), the reduction would imply $c\leq \log m$ where $c$ is the nondeterministic communication complexity of $S\circ g^n$. If merely computing $g$ requires $b$ bits of nondeterministic communication, then clearly $c\geq b$ so that $m\geq 2^b$.

\subsection{A linear lower bound}

\begin{theorem} \label{thm:query-lb}
There is a family of $n$-node bounded-degree graphs $G$ such that the \witness-game for $\TSE_G$ requires query complexity $\Omega(n)$.
\end{theorem}

\paragraph{Relation to \cite{lovasz95search}.}
An analogue of the \eqref{eq:kw-ef} connection holds for query complexity: if there is a deterministic decision tree of height $d$ that solves the search problem $\TSE_G$, we can convert this into a degree-$(d+O(1))$ conical junta for the associated \witness-game. Moreover, if we only have a \emph{randomized} $\epsilon$-error decision tree for the search problem, then the connection gives us a conical junta $h$ that \emph{approximately} solves the \witness-game: $h(z) \in (|\!\viol(z)|-1)\cdot(1\pm\epsilon)$ for all $z$.

Our proof below is robust enough that the $\Omega(n)$ bound holds even for conical juntas that merely approximately solve the \witness-game. Hence we get a randomized $\Omega(n)$ lower bound for $\TSE_G$, which was conjectured by~\cite[p.~125]{lovasz95search}; note however that the earlier work~\cite{goos14communication} already got a near-optimal $\Omega(n/\log n)$ bound. In any case, to our knowledge, this is the first $O(1)$-vs-$\Omega(n)$ separation between certificate complexity and randomized query complexity for search problems.

\paragraph{The proof.}
Fix an $n$-node bounded-degree expander $G=(V,E)$. That is, for any subset $U\subseteq V$ of size $|U|\leq n/2$, the number of edges leaving $U$ is $\Theta(|U|)$. We tacitly equip $G$ with an arbitrary odd-weight node-labeling. Assume for the sake of contradiction that there is a conical junta $h=\sum w_C C$ of degree $o(n)$ for the \witness-game for $\TSE_G$. Let $C$ be a conjunction with $w_C>0$. Denote by $S\subseteq E$ the set of edges that $C$ reads; hence $|S| \leq o(n)$. Below, we write $G\smallsetminus S$ for the graph induced on the edges $E\smallsetminus S$ (deleting nodes that become isolated).

\begin{claim} \label{clm:connected}
We may assume w.l.o.g.\ that $G\smallsetminus S$ is connected.
\end{claim}

\begin{proof}
If $G\smallsetminus S$ is not connected, we may replace $C$ with a conjunction (actually, a sum of them) that reads more input variables; namely, we let $C$ read a larger set of edges $S'\supseteq S$ including all edges from connected components of $G\smallsetminus S$ of ``small'' size $\leq n/2$. When adding some small component $K\subseteq E$ to $S'$ we note that, because $G$ is expanding, the size of $K$ is big-$O$ of the size of the edge boundary of $K$ (which is contained in $S$). On the other hand, every edge in $S$ lies on the boundary of at most two components. It follows that $|S'| = O(|S|)$, i.e., we increased the degree of $h$ only by a constant factor. Now in $G\smallsetminus S'$ we have only components of size $> n/2$, but there can only be one such component.
\end{proof}

\begin{claim} \label{clm:twowitness}
We may assume w.l.o.g.\ that $C$ witnesses at least two fixed nodes with a parity violation (i.e., $C$ reads all the edge labels incident to the two nodes).
\end{claim}
\begin{proof}
Suppose for contradiction that $C$ witnesses at most one violation. Then we may fool $C$ into accepting an input (and hence $h$ into outputting a positive value on that input) where the number of violations is $1$, which is a contradiction to the definition of the \witness-game. Indeed, let $z$ be some input accepted by $C$. Then we may modify $z$ freely on the connected graph $G\smallsetminus S$ (by \autoref{clm:connected}) without affecting $C$'s acceptance: we may eliminate pairs of violations from $z$ by flipping paths (as in \autoref{sec:tseitin-def}) until only one remains. (This is possible since by definition, all the non-witnessed violations of $z$ remain in $G\smallsetminus S$.)
\end{proof}

Let $\mu_i$ ($i$ odd) denote the distribution on inputs that have $i$ violations at a random set of $i$ nodes, and are otherwise random with this property. We may generate an input from $\mu_i$ as follows:
\begin{enumerate}
\item Choose an $i$-set $T_i\subseteq V$ of nodes at random.
\item Let $z\in\Z^E_2$ be any fixed input with $\viol(z)=T_i$.
\item Let $q\in\Z^E_2$ be a random eulerian graph.
\item Output $z+q$.
\end{enumerate}

\autoref{thm:query-lb} follows from the following lemma. Here we identify $C$ with the set (subcube) of inputs it accepts.

\begin{lemma} \label{lem:factor}
$\mu_5(C)\ge(10/3-o(1))\cdot\mu_3(C)$.
\end{lemma}

Indeed, consider the expected output value $\E_{z_i\sim\mu_i}[h(z_i)]$. This should be $2$ for $i=3$, and $4$ for $i=5$, i.e., a factor $2$ increase. However, the above lemma  implies that the output value gets multiplied by more than a factor $3$, which is the final contradiction.

\begin{proof}[Proof of \autoref{lem:factor}]
By \autoref{clm:twowitness} let $\{v_1,v_2\}$ be a pair of nodes where~$C$ witnesses two violations. For $i=3,5$, let $z_i\sim\mu_i$ and denote by $T_i$ the $i$-set of its violations. Then
\begin{align*}
\hspace{2cm}
\mu_3(C)~
&=~\Pr[C(z_3)=1] \\
&=~\Pr[C(z_3)=1\text{ and }T_3\supseteq \{v_1,v_2\} ]\\
&=~\textstyle \binom{n-2}{1}/\binom{n}{3}\cdot\Pr[C(y_3)=1], \tag{for $y_3\coloneqq (z_3\mid T_3\supseteq\{v_1,v_2\})$}\\[5mm]
\mu_5(C)~
&=~\Pr[C(z_5)=1]\\
&=~\Pr[C(z_5)=1\text{ and }T_5\supseteq\{v_1,v_2\}]\\
&=~\textstyle \binom{n-2}{3}/\binom{n}{5} \cdot\Pr[C(y_5)=1]. \tag{for $y_5\coloneqq (z_5\mid T_5\supseteq\{v_1,v_2\})$}
\end{align*}
So their ratio is
\[
\frac{\mu_5(C)}{\mu_3(C)}
~=~\frac{10}{3}
\cdot\frac{\Pr[C(y_5)=1]}{\Pr[C(y_3)=1]}.
\]
Hence the following claim concludes the proof of \autoref{lem:factor}.
\end{proof}

\begin{claim}
$\Pr[C(y_5)=1]/\Pr[C(y_3)=1]\ge 1-o(1)$.
\end{claim}

\begin{proof}
We can generate $y_3$ and $y_5$ jointly as follows:
\begin{itemize}[itemsep=3mm]
\item[\bf\boldmath $y_3$:] Choose $v_3\in V\smallsetminus\{v_1,v_2\}$ uniformly random and let $x_3$ be some input with $\viol(x_3)=\{v_1,v_2,v_3\}$. Output $y_3\coloneqq x_3+q$ where $q$ is a random eulerian graph.
\item[\bf\boldmath $y_5$:] Continuing from the above, choose $\{v_4,v_5\}\subseteq V\smallsetminus\{v_1,v_2,v_3\}$ at random. If possible, let $p$ be a path in $G\smallsetminus S$ joining $\{v_4,v_5\}$ (a ``good'' event), otherwise let $p$ be any path joining $\{v_4,v_5\}$. Output $y_5\coloneqq x_3+p+q$.
\end{itemize}
It suffices to prove the claim conditioned on any particular $v_3$ (and hence also on $x_3$). By \autoref{clm:connected} we have $\Pr[\text{``good''}\mid v_3]=\Pr\bigl[v_4,v_5\in G\smallsetminus S\bigmid v_3\bigr]\ge 1-o(1)$ since $|S|\leq o(n)$. If the ``good'' event occurs, then $C$ cannot distinguish between $y_3=x_3+q$ and $y_5=x_3+p+q$ so that $\Pr[C(y_3)=1\mid v_3]=\Pr\bigl[C(y_5)=1\bigmid\text{``good''},v_3\bigr]$. The claim follows as{\allowdisplaybreaks
\begin{align*}
\Pr[C(y_5)=1\mid v_3]~&\ge~\Pr\bigl[C(y_5)=1\text{ and ``good''}\bigmid v_3\bigr]\\
&=~\Pr\bigl[C(y_5)=1\bigmid\text{``good''},v_3\bigr]\cdot\Pr[\text{``good''}\mid v_3]\\
&=~\Pr[C(y_3)=1\mid v_3]\cdot\Pr[\text{``good''}\mid v_3]\\
&\ge~\Pr[C(y_3)=1\mid v_3]\cdot(1-o(1)).\qedhere
\end{align*}
}\end{proof}

\smallskip
\subsection*{Acknowledgements}

Thanks to Denis Pankratov, Toniann Pitassi, and Robert Robere for discussions. We also thank Samuel Fiorini and Raghu Meka for e-mail correspondence. M.G.\ admits to having a wonderful time at IBM while learning about extended formulations with T.S.\ Jayram and Jan Vondrak.

Part of this research was done while M.G.\ and R.J.\ were attending the \emph{Semidefinite and Matrix Methods for Optimization and Communication} program at the Institute for Mathematical Sciences, National University of Singapore in 2016.
This research was supported in part by NSERC, and in part by the Singapore Ministry of
Education and the National Research Foundation, also through the Tier
3 Grant \emph{Random numbers from quantum processes} MOE2012-T3-1-009.
M.G.\ is partially supported by the Simons Award for Graduate Students in TCS.

\DeclareUrlCommand{\Doi}{\urlstyle{sf}}
\renewcommand{\path}[1]{\small\Doi{#1}}
\renewcommand{\url}[1]{\href{#1}{\small\Doi{#1}}}
\bibliographystyle{alphaurl}
\addcontentsline{toc}{section}{References}

\newcommand{\etalchar}[1]{$^{#1}$}

\end{document}